\documentclass{article}

\usepackage{graphicx}
\usepackage{color}
\usepackage{amsmath}
\usepackage{amsthm}
\usepackage{bm}
\usepackage{amssymb}
\usepackage{booktabs}
\usepackage{graphicx}
\usepackage{algpseudocode}
\usepackage{booktabs}
\usepackage{multirow}
\usepackage{graphicx}
\usepackage{hyperref}
\usepackage{xcolor}
\usepackage{threeparttable}
\usepackage[ruled,norelsize,vlined,linesnumbered]{algorithm2e}
\usepackage{tikz}
\usepackage{filecontents}
\usepackage{multicol}
\usepackage{tcolorbox}
\usepackage{url}
\usepackage{subcaption}
\usepackage{tabularx}
\usepackage{makecell}
\usepackage{ragged2e}
\usepackage{makecell}
\usepackage{array}
\usepackage{enumitem}
\usepackage{footmisc}
\newtheorem{definition} {Definition}

\newtheorem{theorem}    {Theorem}

\usepackage{geometry}
\begin{document}
\title{Structural Leakage in Graph Encryption: Attacks and Defenses}

\author{Hua Shen$^1$(nancy78733@126.com) \quad Renzhi Chen$^1$(13914373958@163.com)\\
        Ge Wu$^2$(gewu@seu.edu.cn) \quad Willy Susilo$^3$(wsusilo@uow.edu.au) \\
        Jing Chen$^4$(chenjing@whu.edu.cn) \quad Mingwu Zhang$^1$(csmwzhang@gmail.com) \\
$^1$ Hubei University of Technology \quad $^2$Southeast University \quad $^3$University of Wollongong \\ $^4$Wuhan University
}

\date{}
\maketitle

\begin{abstract}
Graph encryption schemes (GES) enable secure outsourcing of graph data while supporting efficient queries. This report provides a comprehensive analysis of structural leakage in GES for single-pair shortest path (SPSP) queries, integrating findings from two recent works. First, we analyze PathGES, a scheme designed to resist query recovery attacks through heavy-light decomposition (HLD) and canonical fragment encoding. Our analysis reveals that PathGES suffers from significant imbalances in HLD decomposition, with over 99\% of token-path mappings being one-to-one on real-world datasets, enabling both the Falzon-Paterson attack and side-channel inference of path lengths. Second, we present Fragment Tree attack that exploits these structural weaknesses to recover query contents, achieving up to 10.24\% exact recovery on sparse graphs. Third, we introduce BlindGES, an enhanced scheme incorporating a Merge-and-Divide mechanism and two-level multimap index that reduces one-to-one mappings to below 20\%, cuts setup time by 50\%, reduces storage overhead by 32\%, and limits path length leakage to under 1\%. This report systematically presents attack methodologies, defense mechanisms, security proofs, and experimental evaluations on seven real-world datasets.
\end{abstract}

\vspace{1em}

\newpage
\tableofcontents
\newpage

\section{Introduction}
\subsection{Backgroud}
Graphs can intuitively represent relationships among multiple entities and are widely used in social networks \cite{nettleton2013data, lai2019graphse2, wang2022pegraph}, geographic information systems (GIS) \cite{wu2016privacy, zhu2022application}, financial transactions \cite{henderson2020using}, and communication networks \cite{jiang2022graph, peng2023efficiently}. The rapid growth of graph scale has driven enterprises and users to increasingly outsource graph datasets to cloud servers \cite{selvaraj2023outsourced, wang2023prigsim}. However, graph outsourcing, while operationally advantageous, introduces privacy vulnerabilities. Attacks on compromised infrastructure or untrusted server behaviors could lead to unauthorized disclosure of clients' sensitive information, such as social connectivity patterns, geolocation trajectories, and financial transaction histories.

To address this issue, graphs are typically encrypted before being outsourced to cloud servers \cite{ghosh2021efficient, wang2017secgdb, wang2024cryptgraph}. Graph Encryption scheme (GES) is a specialized form of structured encryption (STE) \cite{meng2015grecs, shen2018cloud, hu2024pruned, ghosh2021efficient, falzon2024PathGES}. STE \cite{chase2010structured} can realize efficient and private server-side queries and belongs to searchable symmetric encryption (SSE) \cite{poh2012structured}. GES first needs to design a specific structure (such as SP-matrix \cite{ghosh2021efficient}, $2$-hop labeling \cite{liu2017graph}, bi-directional index \cite{li2021forward}, Multimap chaining \cite{chase2010structured, poh2012structured, falzon2024PathGES}) to redescribe a graph or to represent the calculation results of a graph. Then, GES encrypts the structure while enabling efficient queries on the encrypted structure. The queries include adjacency relationship queries \cite{yang2024querying, song2024enabling}, graph pattern matching queries \cite{xu2023framework, cao2011privacy}, subgraph matching queries \cite{wang2022oblivgm, guan2023efficient, zuo2022privacy}, shortest distance queries \cite{meng2015grecs, shen2018cloud}, and shortest path queries \cite{xie2016practical, wang2024cryptgraph, ghosh2021efficient, falzon2024PathGES}.

Shortest path queries are fundamental in areas such as navigation, logistics, and network optimization.

\subsection{Evolution of Graph Encryption Schemes}
In 2021, Ghosh, Kamara, and Tamara \cite{ghosh2021efficient} introduced the first graph encryption scheme (GES) supporting single-destination shortest path (SDSP) queries. Their approach constructs a dictionary (SPDX) representing the shortest-path matrix of a graph. A critical limitation is the strict one-to-one mapping between tokens and paths.

In 2022, Falzon and Paterson \cite{falzon2022efficient} exploited this mapping to recover client queries, demonstrating that an adversary with knowledge of the plaintext graph and access to query leakage information can recover the queried node pairs.

In response, in 2024, Falzon et al. Falzon et al. \cite{falzon2024PathGES} proposed PathGES, which associates each token with a canonical segment—a path segment whose length is a power of two. For exponents $\geq$ $2$, one token may correspond to multiple paths, thereby reducing the risk of query recovery. The concept of query recovery was first introduced in \cite{mouratidis2012shortest}, with the core goal of recovering the plaintext queries issued by the client through the exploitation of leakage information generated during the query process.

However, our analysis reveals that PathGES suffers from significant imbalances in the number of heavy and light edges in its HLD decomposition, as well as structural flaws in its token-path mapping architecture. These weaknesses result in a substantial number of one-to-one mappings, making it ineffective against the Falzon-Paterson attack and vulnerable to shortest path length leakage.

\subsection{Contributions of This Report}
This report integrates two complementary works: BreakingGES (attack) exposes critical vulnerabilities in PathGES and develops an attack to exploit them; BlindGES (defense) proposes a defense scheme to address these vulnerabilities.

\textbf{Fragment Tree Attack (BreakingGES)}:
\begin{itemize}
  \item In-depth analysis of PathGES design, revealing that over 99\% of token-path mappings remain one-to-one on real-world datasets
  \item Introduction of Fragment Tree structure for managing path fragments and capturing leakage information
  \item Development of query tree generation algorithm that captures logical relationships between tokens
  \item Rigorous proof of isomorphism between fragment trees and query trees
  \item Experimental validation on seven datasets demonstrating effective query recovery, with up to 10.24\% of queries exactly recovered and an additional 10\% narrowed to 2-5 candidates
\end{itemize}

\textbf{BlindGES Defense (BlindGES)}:
\begin{itemize}
  \item Discovery that at least 84\% of shortest path lengths can be inferred via side-channel attacks on PathGES
  \item Proposal of Merge-and-Divide mechanism to address the mismatch between HLD (a classic tree decomposition algorithm) and its canonical segment encoding, which together cause over 99\% one-to-one token-path mappings on real-world graphs
  \item Design of novel two-level multimap index structure for secure shortest path queries
  \item Experimental demonstration showing that, compared to PathGES, BlindGES reduces one-to-one mappings from over 99\% to below 20\%, limits path length leakage to under 1\%, and cuts setup time by 50\% and storage overhead by 32\%
\end{itemize}

\subsection{Report Organization}
Section 2 presents preliminaries. Section 3 reviews GKT and the Falzon-Paterson attack. Section 4 provides in-depth analysis of PathGES. Section 5 presents the Fragment Tree attack methodology. Section 6 introduces BlindGES defense scheme. Section 7 presents experimental evaluations. Section 8 discusses findings. Section 9 concludes.

\section{Preliminaries}
We use $[n]$ to denote $\{1$, $2$, $\cdots$, $n\}$, $p_1$ $\parallel$ $p_2$ to represent the concatenation of path fragments $p_1$ and $p_2$, and $\lambda$ to denote the security parameter.

\subsection{Data Structures}\label{datastructure}
\textbf{Graph.} A graph $G$ $=$ ($V$, $E$) consists of a non-empty finite vertex set $V$ and a finite edge set $E$. $|V|$ indicates the number of vertices in $G$, and $|E|$ represents the number of edges.

\textbf{Single-Destination Shortest Path Tree (SDSP-Tree)} \cite{frigioni1998semidynamic}. For a given graph $G$, an SDSP-tree records the information of the shortest paths from the other vertices in the graph to the root vertex (the destination). After using a shortest path algorithm (such as the Floyd-Warshall \cite{weisstein2008floyd} algorithm or Dijkstra's algorithm) to find out the shortest paths between all pairs of vertices, we can generate $|V|$ SDSP-trees of $G$. For any $r$ $\in$ $V$, we use $T_r$ to represent the SDSP-tree rooted at $r$.

\textbf{Heavy-Light Decomposition (HLD)} \cite{sleator1981data}. Heavy-light decomposition is a method for decomposing a tree into disjoint subpaths using heavy and light edges in a breadth-first search (BFS) manner. Given a rooted tree $T_r$, the size $size_{T_r}(v)$ of a node $v$ is the number of nodes in the subtree rooted at the node $v$. We use $parent(v)$ to denote the parent node of $v$. If $size_{T_r}(v)$ $>$ $\frac{1}{2} size_{T_r}(parent(v))$, the edge ($v$, $parent(v)$) is a heavy edge, otherwise it is a light edge. For any node in $T_r$, there is at most one heavy edge among the edges connecting it with its children. The input of HLD algorithm is a rooted tree $T_r$ and its output is a disjoint path set $Pset_r$. Fig. \ref{HLD} illustrates the processing procedure of this algorithm through an example.

\textbf{Multimap}. A multimap $M$ is a data structure that is an extension of a dictionary. A traditional dictionary maps each label to a single value; the relationship between label and value is one-to-one. A multimap maps each label $lab$ to a set of values $valSet$; the relationship between label and value is one-to-many. Each label is unique. The multimap structure supports an insert operation that adds a pair ($lab$, $valSet$) into $M$ and a visit operation that, given a label $lab$, returns the associated $valSet$. If $M$ does not hold the search label $lab$, $M[lab]$ returns $\bot$. We denote by $|M|$ the size of $M$ (the number of ($lab$, $valSet$) pairs) and by $|M[lab]|$ the size of $valSet$.

\textbf{Query Tree.} The query tree $Q$ $=$ ($TK$, $E'$) is a tree constructed by the server after collecting all leaked tokens. $TK$ is the set of all leaked tokens in the queries, with each token corresponding to a node of the query tree. The edges of $E'$ represent the connection relationships between tokens. Suppose there exists two mappings $tk_1$ $\mapsto$ $p_1$ and $tk_2$ $\mapsto$ $p_2$, where $p_1$ and $p_2$ are two paths. If these two paths can be linked into a longer path, (i.e., $p_1$ $\parallel$ $p_2$), then there exists a directed edge $(tk_1, tk_2)$ $\in$ $E'$ in the query tree to represent this relationship.

\subsection{Graph Encryption Scheme}
A graph encryption scheme (GES) includes five algorithms \cite{falzon2024PathGES}:

\vspace{-0.5\topsep}
\begin{itemize}
  \setlength{\itemsep}{0pt}
  \setlength{\parsep}{0pt}
  \setlength{\parskip}{0pt}
  \item {\texttt{GES.KeyGen}}: Its input is a security parameter $\lambda$, and its output is a secret key $K$.
  \item {\texttt{GES.Encrypt}}: Its inputs include $K$ and a graph $G$; its output is an encrypted database $ED$.
  \item {\texttt{GES.Token}}: It takes $K$ and a query $q$ and returns the corresponding token $tk$.
  \item {\texttt{GES.Search}}: It uses $tk$ to search $ED$ and returns a response $resp$.
  \item {\texttt{GES.Reveal}}: It takes $resp$ and $K$ and outputs the plaintext of the query result.
\end{itemize}

Except for \texttt{GES.Search}, the client executes all the other algorithms, and the server performs the algorithm \texttt{GES.Search}. GESs realize queries on $ED$ in sublinear time by leaking some information. Describe the information leaked during constructing $ED$ by defining the leakage function $\mathcal{L}_S$, and characterize the information revealed during retrieving $ED$ with the leakage function $\mathcal{L}_Q$. In other words, $\mathcal{L}_S$ and $\mathcal{L}_Q$ define the maximum amount of information that GESs can allow to be leaked. Like \cite{falzon2024PathGES}, we define and prove the security of BlindGES using the real-ideal paradigm concerning a server that learns the outputs of $\mathcal{L}_S$ and $\mathcal{L}_Q$, knows (or even chooses) the graph $G$, and tries to infer the plaintext queries.

\begin{definition}[Adaptive ($\mathcal{L}_S$, $\mathcal{L}_Q$)-secure GES \cite{chase2010structured, ghosh2021efficient, falzon2024PathGES}] \label{Def_1} Let $\mathrm{GES}$ be a graph encryption scheme, $\mathcal{L}_S$ and $\mathcal{L}_Q$ be leakage algorithms. We say that $\mathrm{GES}$ is ($\mathcal{L}_S$, $\mathcal{L}_Q$)-secure against adaptive chosen-query attacks if for any PPT adversary $\mathcal{A}$ and for all $\lambda$ $\geq$ $1$, there exists a PPT $\mathcal{S}$ and a negligible function $negl(\lambda)$ such that
    \begin{equation}		|Pr[\mathrm{\mathbf{Real}}^{\mathrm{GES}}_{\mathcal{A}}(1^{\lambda})=1]-Pr[\mathrm{\mathbf{Ideal}}^{\mathrm{GES}}_{\mathcal{A},\mathcal{S} im}(1^{\lambda})=1]| \leq negl(\lambda). \nonumber
	\end{equation}
and the definitions of games $\mathrm{\mathbf{Real}}^{\mathrm{GES}}_{\mathcal{A}}(1^{\lambda})$ and $\mathrm{\mathbf{Ideal}}^{\mathrm{GES}}_{\mathcal{A},\mathcal{S} im}(1^{\lambda})$ are shown in Figure \ref{Games}.
\end{definition}

\begin{figure}
  \centering
  \begin{minipage}{0.5\textwidth}
      \centering
      \includegraphics[width=\linewidth]{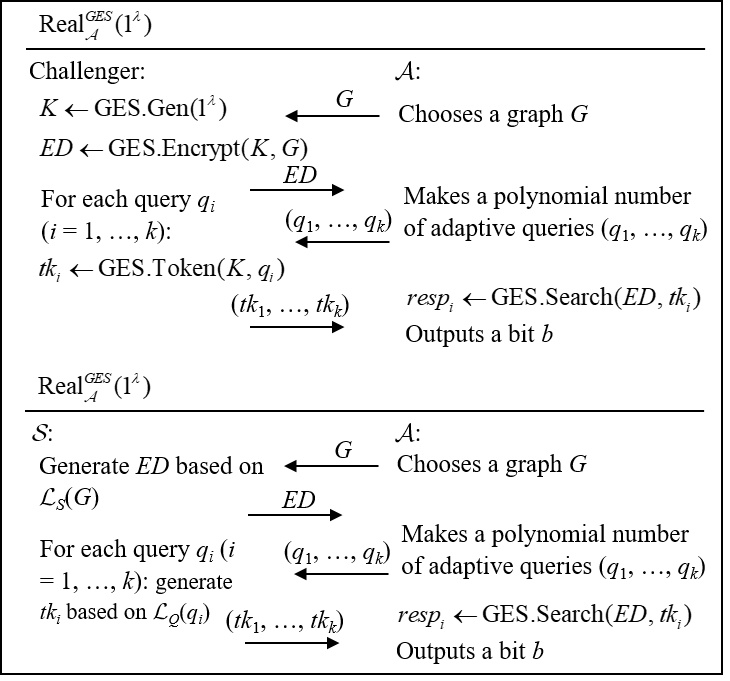}
  \end{minipage}
  \caption{Games $\mathrm{\mathbf{Real}}^{\mathrm{GES}}_{\mathcal{A}}(1^{\lambda})$ and $\mathrm{\mathbf{Ideal}}^{\mathrm{GES}}_{\mathcal{A},\mathcal{S} im}(1^{\lambda})$}
  \label{Games}
\end{figure}

\subsection{Multimap Encryption Scheme}
Multimap encryption schemes fall into response-hiding encryption (EMM-RH) and response-revealed encryption (EMM-RR). An EMM-RH scheme is a tuple of five algorithms \texttt{EMM-RH} $=$ (\texttt{KeyGen}, \texttt{Encrypt}, \texttt{Token}, \texttt{Get}, \texttt{Reveal}), and an EMM-RR is a tuple of four algorithms \texttt{EMM-RR} $=$ (\texttt{KeyGen}, \texttt{Encrypt}, \texttt{Token}, \texttt{Get}):

\begin{itemize}
  \setlength{\itemsep}{0pt}
  \setlength{\parsep}{0pt}
  \setlength{\parskip}{0pt}
  \item {\texttt{EMM-RH.KeyGen}}: It takes as input a security parameter $\lambda$ and outputs a secret key $K$.
  \item {\texttt{EMM-RH.Encrypt}}: It takes as input a key $K$ and a multi-mapping $M$ and outputs an encrypted multimap $EM$.
  \item {\texttt{EMM-RH.Token}}: It takes as input a key $K$ and a label $lab$ and outputs a token $tk$.
  \item {\texttt{EMM-RH.Get}}: It takes as input a token $tk$ and an encrypted multimap $EM$. It outputs a response $resp$.
  \item {\texttt{EMM-RH.Reveal}}: It takes as input a response $resp$ and outputs a set of plaint values $ValSet$.
\end{itemize}

Within the five algorithmic tuples described above, only the \texttt{EMM-RH.Get} algorithm is executed by the server, whereas all remaining algorithms are executed by the client. In contrast to the EMM-RH algorithm suite, the EMM-RR suite omits the reveal algorithm because \texttt{EMM-RR.Get} directly returns $ValSet$. Therefore, the EMM-RR algorithm suite directly exposes the query results to the server.

Like \cite{falzon2024PathGES}, we assume \texttt{EMM-RR} to be ($\mathcal{L}^{EMM-RR}_S$, $\mathcal{L}^{EMM-RR}_Q$)-secure and \texttt{EMM-RH} to be ($\mathcal{L}^{EMM-RH}_S$, $\mathcal{L}^{EMM-RH}_Q$)-secure. The setup leakages $\mathcal{L}^{EMM-RR}_S(M)$ and $\mathcal{L}^{EMM-RH}_S(M)$ are $|M|$ and $\sum_{i=1}^{|M|}|ValSet_i|$. The query leakage $\mathcal{L}^{EMM-RR}_Q$($M$, $lab_1$, $\cdots$, $lab_k$) is ($QP^{EMM-RR}$, $AP^{EMM-RR}$), and the query leakage $\mathcal{L}^{EMM-RH}_Q$($M$, $lab_1$, $\cdots$, $lab_k$) is
($QP^{EMM-RH}$, $Vol^{EMM-RH}$). $QP^{EMM-RR}$($M$, $lab_1$, $\cdots$, $lab_k)$ and $QP^{EMM-RH}$($M$, $lab_1$, $\cdots$, $lab_k)$ are the query pattern that reveals whether two queries are equal. The pattern is a $k$ $\times$ $k$ matrix of bits. If $lab_i$ $=$ $lab_j$, a $1$ is set at location ($i$, $j$) of the matrix. $AP^{EMM-RR}$($M$, $lab_1$, $\cdots$, $lab_k$) $=$ $\{M[lab_1]$, $\cdots$, $M[lab_k]\}$ is the access pattern that reveals the individual values returned for each query. $Vol^{EMM-RH}$($M$, $lab_1$, $\cdots$, $lab_k$) $=$ $\sum_{i=1}^{k}|M[lab_i]|$ is the volume pattern that reveals the number of the values returned for each query.

\subsection{Tree Isomorphism and Canonical Names}\label{TreeIsoCanNames}
\textbf{Graph Isomorphism} Let $G_1$ $=$ ($V_1$, $E_1$) and $G_2$ $=$ ($V_2$, $E_2$) be two graphs. If there exists a bijection $\phi$: $V_1$ $\to$ $V_2$ such that for any $u$, $v$ $\in$ $V_1$, we have $\phi(u)$, $\phi(v)$ $\in$ $V_1$, and if ($u$, $v$) $\in$ $E_1$, we have $(u, v)$ $\in$ $E_1$ $\iff$ ($\phi(u)$, $\phi(v)$) $\in$ $E_2$, then $G_1$ and $G_2$ are \emph{isomorphic}, denoted $G_1$ $\cong$ $G_2$. In this case, $\phi$ is called an \emph{isomorphism} between the two graphs. Let $T_{r_1}$ $=$ ($V_{r_1}$, $E_{r_1}$) and $T_{r_2}$ $=$ ($V_{r_2}$, $E_{r_2}$) be rooted trees. An \emph{isomorphism} between $T_{r_1}$ and $T_{r_2}$ is a graph isomorphism $\varphi$: $V_{r_1}$ $\to$ $V_{r_2}$ such that $\varphi(r_1)$ $=$ $r_2$.

\textbf{Canonical Name.} The Falzon--Paterson attack uses canonical names to encode each node in a tree, reflecting the structure of its subtree. Canonical names are built recursively in a bottom-up manner: the canonical name of a leaf node is designated as the fixed string {\tt "10"}. For an internal node, the first step is to sort the canonical names of all its child nodes in ascending order based on their length. These sorted names are then concatenated with the delimiters {\tt "1"} and {\tt "0"} added at the beginning and end, respectively. The formal method for calculating a canonical name is $\mbox{Name}(v)$ $=$ $\verb+"+1$ $||$ $\mbox{children\_names}$ $||$ $0\verb+"+$. This naming method depends solely on the subtree structure, ensuring that structurally identical subtrees yield identical canonical names.

\textbf{PathName.} In a rooted tree $T_r$ with $r$ as the root, a path name associated with a node $v$ (denoted as $PathName(v)$) encodes the path structure from $v$ to the root node $r$. To construct it, we first calculate the canonical name of each node. Then, we perform a depth-first traversal from the root node, using a stack to determine the path names for all nodes. For the root $r$, $\mbox{PathName}(r)$ $=$ $h(\mbox{Name}(r))$; for a non-root node $v$, suppose its parent is $u$, $\mbox{PathName}(v)$ $=$ $h(\mbox{Name}(v)$ $||$ $\mbox{PathName}(u))$. Here, $h(\cdot)$ denotes a hash function (e.g., truncated SHA-256) used to compress the length of the path name to $\mathsf{O}(\log n)$. This construction ensures that $\mbox{PathName}(v)$ encodes both a local structure of $T_r$ and the information of the path from $v$ to $r$. Therefore, such path names are suitable for structural matching and encrypted query recovery tasks.

\section{GKT Scheme Review and Falzon-Paterson Attack}
\subsection{GKT Scheme Overview}
In 2021, Ghosh, Kamara, and Tamara \cite{ghosh2021efficient} introduced the first graph encryption scheme (GES) supporting single-destination shortest path (SDSP) queries. Their approach constructs a dictionary (SPDX) representing the shortest-path matrix (SP-matrix) of a graph, where each entry $\mathrm{SPDX}[(u, v)] = (w, v)$ indicates that the next hop from $u$ to $v$ is $w$. To enable recursive traversal on encrypted data, the structure is modified to $\mathrm{SPDX^{\prime}}$, where each value becomes a tuple $(tk_{(w, v)}, ct_{(w, v)})$ containing a token and a ciphertext.

A distinctive feature of GKT is that the encrypted structure enjoys the same asymptotic setup time, query time, and space as the plaintext SP matrix.

\subsection{Core Vulnerability of GKT}
A critical limitation is the strict one-to-one mapping between tokens and paths. Suppose the query request is ($u$, $v$), and the shortest path from $u$ to $v$ is ($u$, $w_0$, $w_1$, $w_2$, $w_3$, $v$). The client sends the token $tk_{(u,v)}$ corresponding to ($u$, $v$) to the server. The server retrieves the encrypted SP-matrix and will successively obtain $tk_{(w_0, v)}$, $tk_{(w_1, v)}$, $tk_{(w_2, v)}$, $tk_{(w_3, v)}$, and $tk_{(v, v)}$. As a result, the server can observe all the tokens involved in a query process as well as the relationships between these tokens.

\subsection{Falzon-Paterson Attack}
Falzon and Paterson \cite{falzon2022efficient} developed a query recovery attack that exploits the one-to-one mapping between tokens and paths in GKT. The attack consists of four phases: precomputation, query observation and query tree construction, mapping dictionary building, and query recovery.

\subsubsection{Precomputation (Attack Preparation)}
Before observing any client queries, the server precomputes necessary information based on the plaintext graph $G$:
\begin{itemize}
  \item \textbf{Build SDSP trees}: For each vertex $r \in V$, construct the SDSP tree $T_r$ rooted at $r$, representing the shortest paths from all other vertices to $r$.
  \item \textbf{Compute PathNames}: For each node in every SDSP tree $T_r$, calculate its PathName. As defined in Section \ref{TreeIsoCanNames}, PathName encodes the path structure from the node to the root, computed recursively using canonical names.
  \item \textbf{Construct mapping table $M$}: Build a mapping table $M$ indexed by PathNames.  For each node $u$ in $T_r$ with PathName $pn$, the entry $M[pn]$ stores the candidate query pair ($u$, $r$). Queries corresponding to nodes with the same PathName are grouped together in the same row of $M$.
\end{itemize}
\noindent This precomputation is performed once before any attack execution.

\subsubsection{Query Observation and Query Tree Construction}
After responding to a sufficient number of query requests from the client, the server constructs a query tree that represents the query structure observed during the query process.

For a fixed destination vertex $v$, the server constructs query tree $Q_v$ as follows:
\begin{itemize}
  \item The root of $Q_v$ is the token $tk_{(v, v)}$
  \item If a token $tk_{(w, v)}$ is observed to be followed by $tk_{(v, v)}$ in a query sequence, a directed edge is added from $tk_{(w, v)}$ to the root
  \item Recursively, if $tk_{(u, v)}$ is followed by $tk_{(w, v)}$, a directed edge is added from $tk_{(u, v)}$ to $tk_{(w, v)}$
\end{itemize}

\subsubsection{Building Mapping Dictionary}
After the query tree is built, the server calculates the path name of each node in the query tree. It records the node value (i.e., the token) and its corresponding PathName in a mapping dictionary $D$, where the token serves as the key of the dictionary and the PathName as the value (i.e., $D:$ token $\rightarrow$ PathName).

\subsubsection{Query Recovery}
The server uses both $M$ (from precomputation) and $D$ (from query observation) to launch the query recovery attack:
\begin{itemize}
  \item For a target query token $tk$, retrieve its PathName from $D: pn = D[tk]$
  \item Use this PathName to query the mapping table $M$, obtaining the set of candidate query pairs: $M[pn]$
  \item The candidate set contains all possible ($u$, $v$) pairs whose corresponding node in the SDSP tree shares the same PathName
\end{itemize}

Fig. \ref{GKT-attack} illustrates an example of the Falzon-Paterson attack. The strings marked in blue in the figure are the canonical names corresponding to the nodes. When the server receives the query token $tk_{(5,1)}$, it retrieves $D$ to obtain the corresponding path name $h(10 || PathName(4))$. Then, the server retrieves $M$ based on $h(10 || PathName(4))$ to get the candidate query (5,1) that matches the query token $tk_{(5,1)}$. Since there is only one candidate query, the server can determine that the query corresponding to $tk_{(5,1)}$ is this candidate query. When the server receives the query token $tk_{(8,1)}$, it retrieves $D$ to obtain the corresponding path name $h(10 || PathName(1))$. Subsequently, the server retrieves $M$ according to $h(10 || PathName(1))$ and finds that the candidate queries matching the query token $tk_{(8,1)}$ are ($8$, $1$) and ($7$, $1$), narrowing the target range to a small candidate set. From the above analysis, it can be found that the Falzon-Paterson attack does not break the underlying encryption primitives; rather, it solely leverages the structural leakage inherent in GKT to recover query pairs or precisely reduce them to a minimal candidate set.

\begin{figure}[h]
  \centering
  \begin{minipage}{0.7\textwidth}
      \centering
      \includegraphics[width=\linewidth]{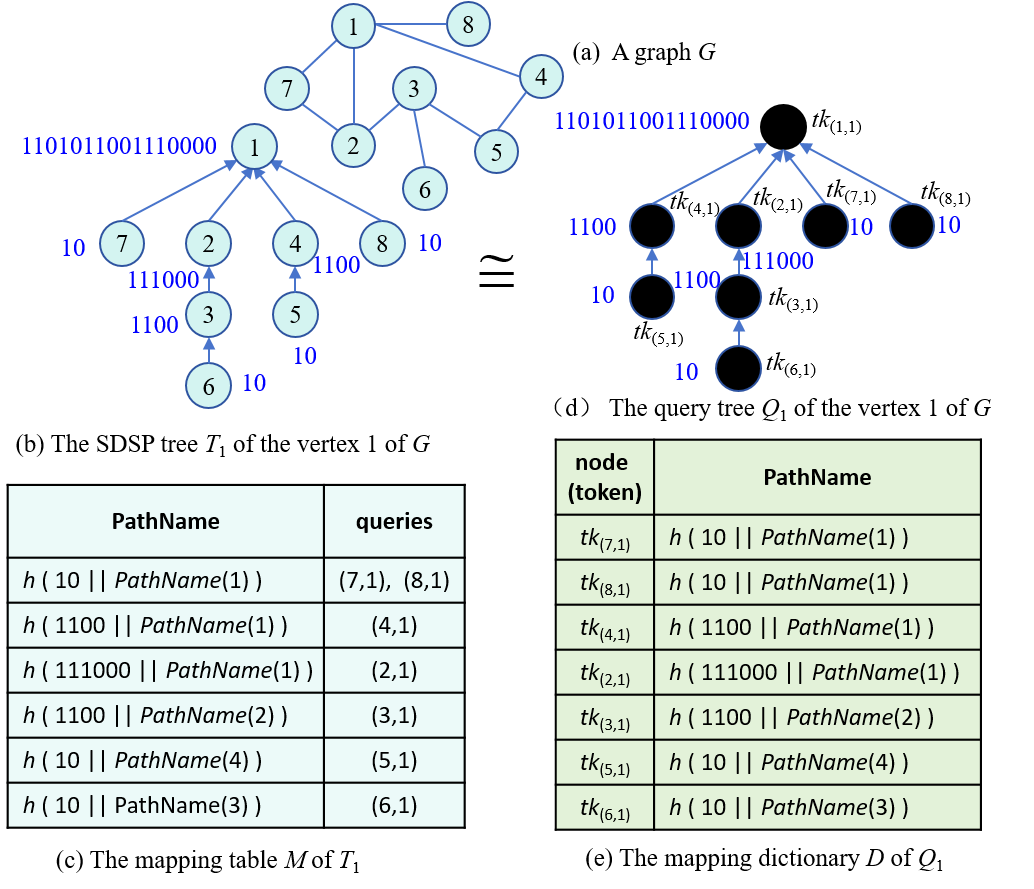}
  \end{minipage}
  \caption{An example of the Falzon-Paterson attack.}
  \label{GKT-attack}
\end{figure}

\textbf{Attack Effectiveness}. According to \cite{falzon2022efficient}, the Falzon-Paterson attack can directly recover up to 21.9\% of query plaintexts in GKT, and for approximately 50\% of the queries, the candidate set can be reduced to no more than three options. These results are based on experiments conducted on the p2p-Gnutella08 and p2p-Gnutella04 datasets.

\section{In-depth Analysis of PathGES}
\subsection{Overview of PathGES}
PathGES \cite{falzon2024PathGES} can be viewed as an enhanced version of GKT\cite{ghosh2021efficient} designed to resist the Falzon-Paterson attack \cite{falzon2022efficient}. In GKT, there exists a strict one-to-one mapping between tokens and paths, allowing the server to observe all tokens and their relationships during the query process. The Falzon-Paterson attack exploits this one-to-one mapping to recover query contents.

The core problem that PathGES needs to solve is: design a mechanism that breaks the unique binding between tokens and paths while still enabling correct recovery of shortest paths. To achieve this goal, PathGES employs the following core techniques:
\begin{itemize}
  \item Decomposing shortest paths using the Heavy-Light Decomposition (HLD) algorithm
  \item Normalizing path fragments into canonical segments (lengths that are powers of two)
  \item Designing a labeling method for normalized fragments
  \item Constructing a two-level multimap index structure
\end{itemize}

\subsubsection{Overall Architecture}
For a given graph $G$, PathGES first needs to construct the SDSP forest $\{T_r\}_{r \in V}$ and decomposes each SDSP tree $T_r$ into disjoint paths $PSet_r$ using HLD. After that, PathGES encodes each path in $PSet_r$ and converts it into a set of canonical segments. Following this, PathGES generates a token for each canonical segment and constructs a two-level index using two multimaps, $M_1$ and $M_2$:
\begin{itemize}
  \item $M_2$ is used to retrieve canonical segments
  \item $M_1$ is used to retrieve tokens that search $M_2$
\end{itemize}

PathGES uses the tokens corresponding to SDSP queries to retrieve $M_1$. Finally, PathGES encrypts $M_1$ and $M_2$ and stores them on the server.

\subsubsection{HLD-based Path Decomposition}
For relevant definitions of HLD, please refer to Section \ref{datastructure}. Given a graph $G$, PathGES first needs to construct an SDSP tree for each vertex in the graph. Suppose the SDSP tree $T_v$ corresponding to vertex $T_v$ is as shown in Fig. \ref{HLD}. Fig. \ref{HLD} illustrates the process of fragmenting the shortest paths represented by $T_v$ using the HLD algorithm, resulting in a set of path fragments denoted as $PSet_v$. The path fragments in $PSet_v$ are mutually disjoint. First, PathGES calculates the size of the subtree for each node in $T_v$, then distinguishes between heavy edges and light edges within $T_v$ and performs fragmentation processing on the paths to obtain the corresponding set of fragmented paths $PSet_v$. If the currently visited node is connected to its child via a light edge, this light edge is pushed into a queue; if the currently visited node is connected to its child via a heavy edge, the current path fragment continues to grow downward along this heavy edge. Otherwise, the growth stops, and the path fragment is merged into $PSet_v$. At this point, it is determined whether the queue is empty. If it is empty, the fragmentation process ends; otherwise, an edge is removed from the queue, and the above process is repeated with the next node in the edge as the currently visited node.

\begin{figure*}[!htbp]
  \centering
  \begin{minipage}{1.0\textwidth}
      \centering
      \includegraphics[width=1.0\linewidth]{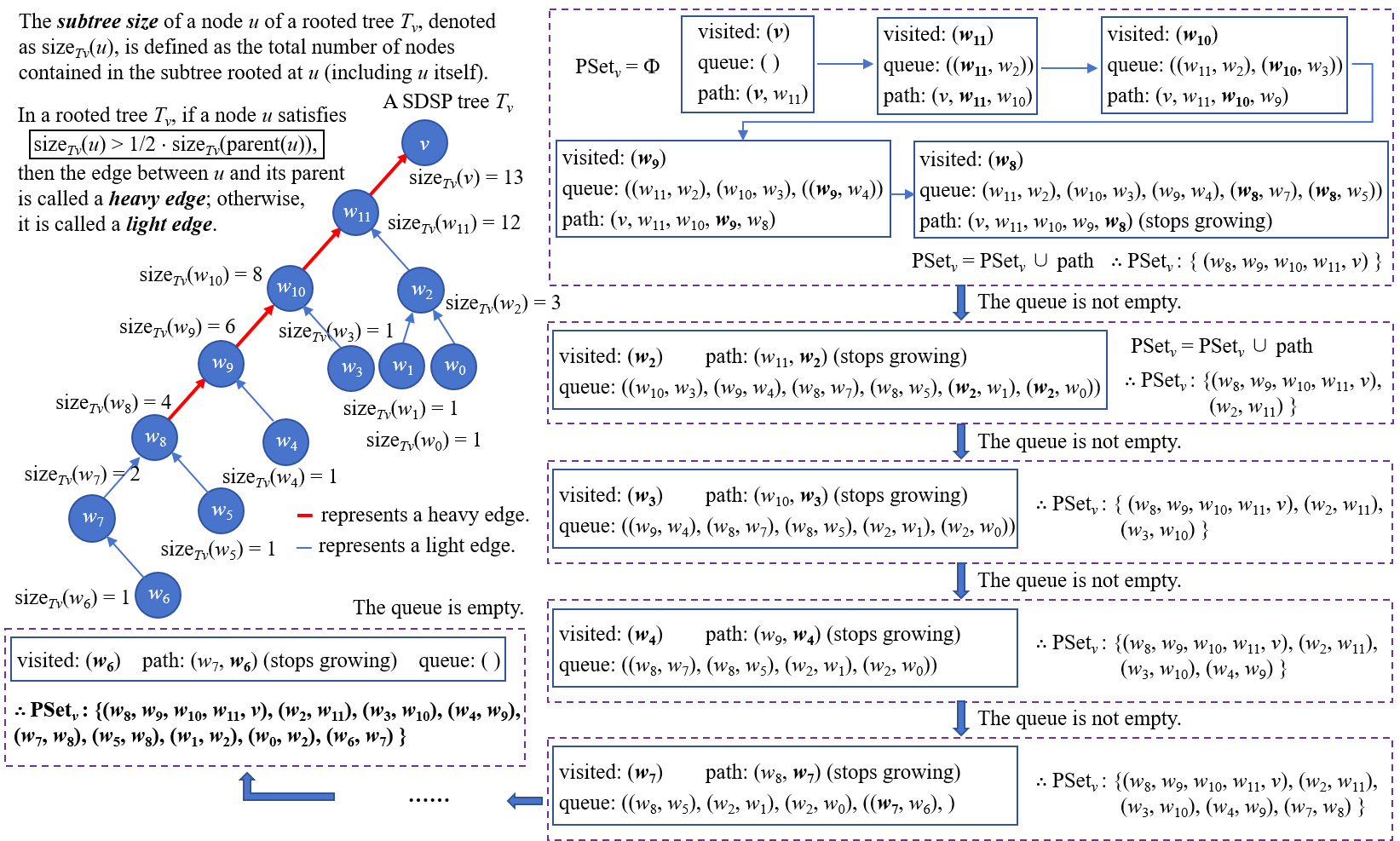}
      \caption{An example of the HLD algorithm.}
      \label{HLD}
    \end{minipage}
\end{figure*}

\subsubsection{Canonical Segment Encoding and Label Assignment}\label{CanSegEncoding}
After obtaining the disjoint path set $PSet_v$ through HLD decomposition, PathGES normalizes each path fragment by padding it with dummy nodes to make its length a power of two. A path fragment whose length is a power of two is called a \emph{canonical fragment}.

\textbf{Example.} suppose ($u'$, $w'_4$, $w'_3$, $w'_2$, $w'_1$, $v'$) is a disjoint path of $T_v$, where $v'$ is the node among these nodes that is closest to the root $v$ of $T_v$. In this path fragment, the edge ($w'_1$, $v'$) may be a light or heavy edge, and others are heavy. This path is encoded into 4 canonical segments: ($w'_1$, $v$), ($w'_2$, $w'_1$, $v$), ($w'_4$, $w'_3$, $w'_2$, $w'_1$, $v$), and ($p$, $p$, $u$, $w'_4$, $w'_3$, $w'_2$, $w'_1$, $v$), where `$p$' is a padding node. These segments are stored as the values of the label/value pairs in $M_2$, and the label determines the storage location. The labels of the four fragments are: ($r$, $u$, $v$, $0$), ($r$, $u$, $v$, $1$), ($r$, $u$, $v$, $2$), ($r$, $u$, $v$, $3$), respectively. Their tokens are $tk_{(r, u, v, 0)}$, $tk_{(r, u, v, 1)}$, $tk_{(r, u, v, 2)}$, and $tk_{(r, u, v, 3)}$.

\textbf{One-to-Many Mapping.} The key innovation of PathGES is that a token can correspond to multiple paths:
\begin{itemize}
  \item $tk_{(r, u, v, 0)}$ corresponds to ($w'_1$, $v$)
  \item $tk_{(r, u, v, 1)}$ corresponds to ($w'_2$, $w'_1$, $v$)
  \item $tk_{(r, u, v, 2)}$ corresponds to two paths: ($w'_3$, $w'_2$, $w'_1$, $v$) and ($w'_4$, $w'_3$, $w'_2$, $w'_1$, $v$)
  \item $tk_{(r, u, v, 3)}$ corresponds to ($u$, $w'_4$, $w'_3$, $w'_2$, $w'_1$, $v$)
\end{itemize}

\noindent Fig. \ref{Normalizing} illustrates the process of normalizing and assigning labels to this path fragment. Generally, the label $(r, u, v, j)$ represents a canonical fragment in the rooted tree with root $r$, where the fragment has a length of $2^j$, starts at $u$, and ends at $v$. $tk_{(r, u, v, j)}$ can correspond to at most $2^j - 2^{j-1}$ paths. The relationship between tokens and paths is no longer necessarily one-to-one, which is intended to help PathGES resist the Falzon-Paterson attack.

\begin{figure}[h]
  \centering
  \begin{minipage}{0.7\textwidth}
      \centering
      \includegraphics[width=1.0\linewidth]{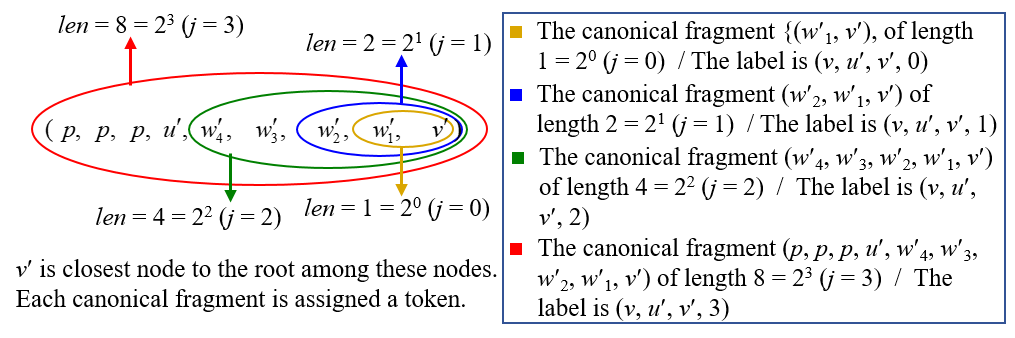}
  \end{minipage}
  \caption{An example of normalizing and assigning labels.}
  \label{Normalizing}
\end{figure}

\subsubsection{Two-Level Index Structure}
PathGES constructs a two-level index structure consisting of two mapping tables, $M_1$ and $M_2$:

\textbf{$M_2$ (Second Level).} The canonical segments are stored in $M_2$. The token for retrieving the canonical fragment with the label $(r, u, v, j)$ from $M_2$ is $M_2$. Note that $M_2$ achieves a one-to-many mapping from tokens corresponding to canonical fragments to paths (see Section \ref{CanSegEncoding}). A token $tk_{(r, u, v, j)}$ can correspond to multiple paths.

\textbf{$M_1$ (First Level).} For the SDSP query ($u$, $v$),  PathGES stores the search tokens of the canonical fragments related to the query in the position indicated by ($u$, $v$) in $M_1$. The token for retrieving these related tokens from $M_1$ is $tk_{(u, v)}$. Note that $M_1$ achieves a one-to-many mapping from SDSP query tokens to tokens corresponding to canonical fragments. A token $tk_{(u, v)}$ can correspond to multiple canonical fragment tokens.

\noindent Fig. \ref{EM1_EM2} shows an example of the two-level index structure of PathGES. $M_1$ and $M_2$ of Fig. \ref{EM1_EM2} correspond to the example in Fig. \ref{HLD}. To ensure the non-interactive way of queries, PathGES employs a response-revealing multimap encryption (EMM-RR) to encrypt $M_1$ and a response-hiding multimap encryption (EMM-RH) to encrypt $M_2$. The encrypted $M_1$ and $M_2$ are denoted as $EM_1$ and $EM_2$, respectively.

\begin{figure}[h]
  \centering
  \begin{minipage}{0.7\textwidth}
      \centering
      \includegraphics[width=1.0\linewidth]{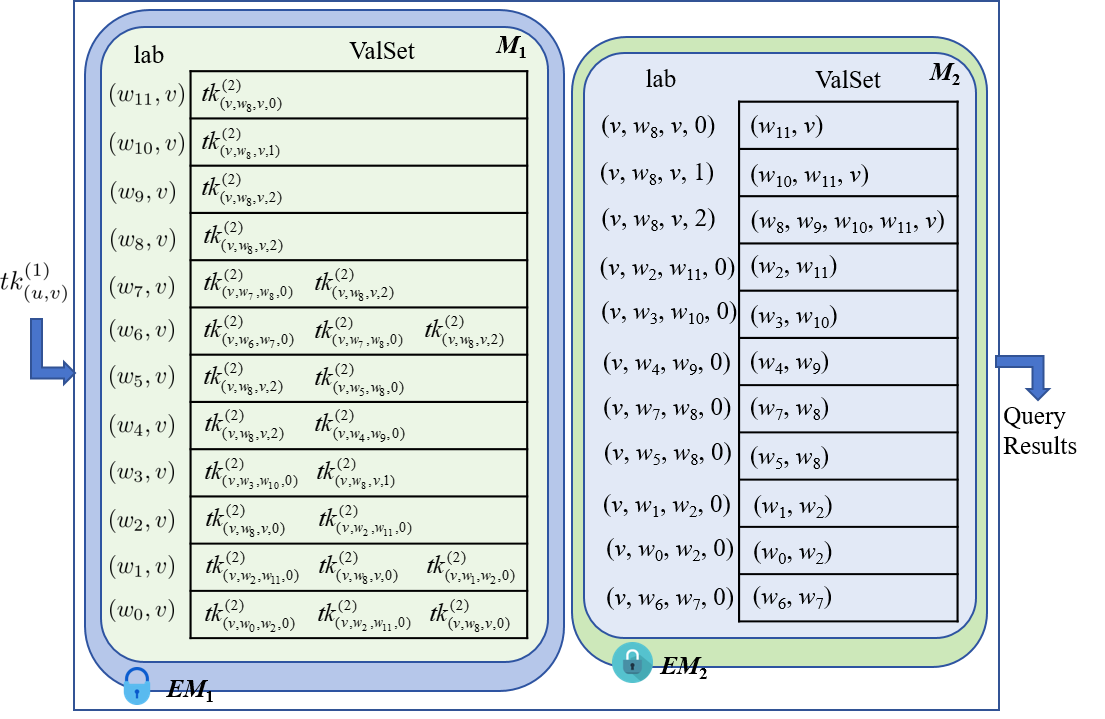}
      \caption{An example of the two-level index structure of PathGES.}
      \label{EM1_EM2}
    \end{minipage}
\end{figure}

\subsubsection{Query Processing}
The query processing in PathGES proceeds as follows:
\begin{itemize}
  \item \textbf{Token Generation.} The client generates a token $tk_{(u,v)}$ for the SPSP query ($u$, $v$) using the secret key.
  \item \textbf{Search.} The server uses $tk_{(u,v)}$ to retrieve from $M_1$ the set of second-level tokens $\{tk_{(r,u,v,j)}\}$ associated with the query.
  \item \textbf{Fragment Retrieval.} The server then uses each second-level token to retrieve the corresponding canonical segment from $M_2$.
  \item \textbf{Response.} The server returns all retrieved canonical segments to the client.
  \item \textbf{Reveal.} The client reconstructs the shortest path by concatenating the canonical segments in the correct order.
\end{itemize}

\subsubsection{Defense Mechanism}
PathGES aims to mitigate the structural leakage risks inherent in GKT, where tokens are in one-to-one correspondence with paths. The defense mechanism of PathGES relies on two key insights:

\begin{itemize}
  \item 1. \textbf{Path Fragmentation.} By decomposing shortest paths into smaller fragments using HLD, PathGES ensures that no single token represents an entire path.
  \item 2. \textbf{Canonical Segment Encoding.} By normalizing fragments to powers-of-two lengths, PathGES enables one-to-many mappings between tokens and paths, obscuring the exact path structure from the server.
\end{itemize}

\textbf{Example of Defense.} Consider that the shortest path corresponding to an SDSP query is decomposed via HLD and normalized, yielding two canonical fragments of lengths $1$ and $4$ (e.g., the SDSP query ($w_7$, $v$) in Fig. \ref{EM1_EM2}). $tk_{(w_7,v)}$ corresponds to two canonical fragment tokens $tk_{(r,w_7,w_8,0)}$ and $tk_{(r,w_8,v,2)}$. A canonical fragment token for $j = 2$ may correspond to a path of length 3 or a path of length 4, partially obscuring the path structure. The server cannot distinguish whether the original path had $3$ edges or $4$ edges.

\subsection{Structural Leakage Problems}
\subsubsection{HLD Imbalance}
Our analysis across six datasets (Table \ref{tab:countshortpaths}) reveals that over $98\%$ of decomposed subpaths have length $1$ or $2$ (each being a canonical segment itself) and $99\%$ of decomposed subpaths have length $1$, $2$, and $3$ (each also being a canonical segment itself). This implies that in PathGES, at least $99\%$ of tokens map to only one segment. As a result, if a shortest path consists solely of such short segments, its total length is directly exposed.

\begin{table*}[ht]
    \centering
    \footnotesize
    \caption{Statistics on path fragments lengths}
    \label{tab:countshortpaths}
    \begin{threeparttable}
        \begin{tabular}{|c|c|c|c|c|c|c|}
            \hline
            Data & internetRouting & Ca-GrQc & email-EU-core & facebook-combined & p2p-Gnutella08 & p2p-Gnutella04 \\
            \hline
            $|V|$            & 35    & 46    & 1,005   & 4,309      & 6,301      & 10,876   \\
            $|E|$            & 323   & 1,030 & 16,706  & 88,234     & 20,777     & 39,994   \\
            Density          & 0.543 & 0.995 & 0.02    & 0.011      & 0.001      & 0.0006     \\
            Total Fragments  & 1,141 & 2,069 & 878,147 & 16,281,762 & 34,619,060 & 101,440,282\\
            ShortFragment I  & 1,141 & 2,069 & 875,638 & 16,276,777 & 33,997,988 & 99664466   \\
            Percentage       & 100\% & 100\% & 99.71\% & 99.97\%    & 98.21\%    & 98.25\%     \\
            ShortFragment II & 1,141 & 2,069 & 877,940 & 16,281,560 & 34,578,445 & 101,354,239\\
            Percentage       & 100\% & 100\% & 99.98\% & 99.99\%    & 99.88\%    & 99.92\%    \\
            \hline
        \end{tabular}
        \begin{tablenotes}
            \footnotesize
            \item[1] $d = \frac{2|E|}{|V|^2 - |V|}$ represents the density of the graph.
            \item[2] ShortFragment I refers to these short fragments with lengths of 1 and 2.
            \item[3] ShortFragment II refers to these short fragments with lengths of 1, 2, and 3.
            \item[4] Percentage indicates the proportion of the short fragments among all fragments.
        \end{tablenotes}
    \end{threeparttable}
\end{table*}

\subsubsection{Side-Channel Leakage}
PathGES causes additional information leakage during the process of using a ($\mathcal{L}^{EMM-RH}_S$, $\mathcal{L}^{EMM-RH}_Q$)-secure EMM-RH scheme, where $\mathcal{L}^{EMM-RH}_Q$ $=$ ($QP^{EMM-RH}$, $Vol^{EMM-RH}$), $Vol^{EMM-RH}(M_2, lab)$ $=$ $|M_2[lab]|$. The value of $M_2[lab]$ is the ciphertext of a canonical path fragment. The canonical fragment encoding method used in PathGES makes the length of the path fragments stored in $M_2$ exhibit a regularity, that is, being a power of $2$. Therefore, the size of the corresponding ciphertext also shows a certain regularity. By observing all the ciphertexts in $M_2$ or listening to multiple responses, an adversary can consider the smallest ciphertext size as the ciphertext size corresponding to a path of length $1$. Then, by observing the target response, the adversary can infer the path length of the query path based on the size of the ciphertext contained in the target response. The standardization (i.e., the regularity shown) of path fragments brought about by the canonical fragment encoding makes it possible to conduct side-channel attacks on PathGES. We have verified this through experiments as well. This information leakage is not defined in $\mathcal{L}^{PathGES}_Q$. Therefore, we consider that PathGES has vulnerabilities to side-channel attacks. The Merge-and-Divide mechanism of BlindGES makes the lengths of the fragments stored in $M_2$ be fixed at two lengths, $l_1$ and $l_2$. Therefore, BlindGES will not cause additional information leakage during the use of a EMM-RH scheme.

\subsubsection{Subpath Structure Leakage}
PathGES pads a path with a length of $l$ to a path of $2^{\lfloor \log n \rfloor + 1}$. Then, $\lfloor \log n \rfloor + 2$ canonical path fragments with lengths of $2^j$ ($j$ $=$ $0$, $1$, $\cdots$, $\lfloor \log n \rfloor + 1$) are extracted, every time starting from the vertex (assumed to be $v$) closest to the root on this path. Suppose the length of this path is $2^{\lfloor \log n \rfloor}$ $+$ $k$, where $1$ $\leq$ $k$ $\leq$ $2^{\lfloor \log n \rfloor + 1}$ $-$ $2^{\lfloor \log n \rfloor}$. The ($2^{\lfloor \log n \rfloor}$ $+$ $j$)-th vertex on this path is $v_j$, ($j$ $=$ $1$, $\cdots$, $k$), and assume that the token corresponding to the canonical path fragment with a length of $2^{\lfloor \log n \rfloor + 1}$ for this path is $tk$. By retrieving $M_1$, it is found that the token sets corresponding to the SPSP queries ($v$, $v_j$) ($j$ $=$ $1$, $\cdots$, $k$) all contain $tk$. For the query ($v$, $v_1$), it obtains additional information: the shortest path information of ($v$, $v_{j^{\prime}}$) ($j^{\prime}$ $=$ $2$, $\cdots$, $k$). PathGES uses the multi-path fragment information contained in the canonical fragments to achieve a one-to-many mapping between tokens and paths. However, the canonical path fragments contain multiple path fragments with a sub-path structure (i.e., short paths are sub-paths of long paths), leading to additional information leakage.

\subsection{Summary}
PathGES fails to achieve its intended design goals. The prevalence of non-highly imbalanced trees in real-world graphs prevents generation of sufficiently long canonical fragments. Over $99\%$ of tokens map to a single unique path fragment, creating a critical vulnerability that enables query recovery attacks.

\section{Our Attack: Fragment Tree Attack Methodology}
\subsection{Threat Model and Assumptions}
This section defines the threat model and assumptions underlying the Fragment Tree attack.

We assume the server acts as an honest-but-curious adversary. While honestly answering SDSP queries, it also attempts to recover the underlying query content by leveraging available leakage. The attack aims to infer the plaintext values of executed queries based on the given graph and query-induced leakage.

Following the security assumptions adopted in the Falzon-Paterson attack \cite{falzon2022efficient} and in PathGES \cite{falzon2024PathGES}, we assume the graph is public. That is, the server knows the original graph that was encrypted to produce $EM_1$ and $EM_2$. As the adversary, the server can acquire all initial search tokens. Furthermore, since PathGES instantiates $M_1$ using a response-revealing encrypted multimap scheme, the server can reconstruct $M_1$ after observing a sufficient number of initial token queries. We also assume the server is familiar with standard shortest path algorithms and the HLD algorithm.

The adversary aims to: 1) \textbf{Query Recovery.} For each observed query token $tk$, recover the corresponding source-destination pair ($u$, $v$) that the client queried. 2) \textbf{Approximate Query Recovery.} If exact recovery is impossible, narrow the candidate set to as few possibilities as possible (ideally 2-5 candidates). 3) \textbf{Path Length Inference.} Infer the length of the shortest path from the response size (side-channel attack).

\subsection{Overview of Attack Process}
The basic process of our attack is as follows: First, the server constructs an SDSP tree for each vertex based on a given graph, and executes the HLD algorithm on each SDSP tree to obtain the corresponding $PSet$ (these two steps are the same as those in PathGES). Then, the server constructs fragment trees introduced in this paper based on $PSet$, while building query trees based on the leaked information of PathGES (for example, $M_1$). Finally, the query recovery attack on PathGES is implemented by leveraging the isomorphic relationships between fragment trees and query trees.

For example, Fig. \ref{FragmentTree} shows the fragment trees constructed based on $PSet_v$ derived from Fig. \ref{HLD}, as well as the query trees built from $M_1$ in Fig. \ref{EM1_EM2}. If they are isomorphic, we can implement the following query recovery attack. By comparing the topology between the fragment trees and the query trees, we can determine that $tk^{(2)}_{(v, w_8, v, 1)}$ maps to the path fragment ($w_{10}$, $w_{11}$, $v$), $tk^{(2)}_{(v, w_3, w_{10}, 0)}$ maps to ($w_3$, $w_{10}$), $tk^{(2)}_{(v, w_8, v, 0)}$ maps to ($w_{11}$, $v$), $tk^{(2)}_{(v, w_2, w_{11}, 0)}$ maps to ($w_2$, $w_{11}$), $tk^{(2)}_{(v, w_8, v, 2)}$ maps to ($w_8$, $w_9$, $w_{10}$, $w_{11}$, $v$), $tk_{(v, w_7, w_8, 0)}^{(2)}$ maps to ($w_7$, $w_8$), and $tk^{(2)}_{(v, w_6, w_7, 0)}$ maps to ($w_6$, $w_7$). However, since the tokens $tk^{(2)}_{(v, w_1, w_2, 0)}$ and $tk_{(v, w_0, w_2, 0)}^{(2)}$ share the same topology, we cannot accurately distinguish their mappings to the path fragments ($w_1$, $w_2$) and ($w_0$, $w_2$). Similarly, the mappings between $tk^{(2)}_{(v, w_4, w_9, 0)}$, $tk^{(2)}_{(v, w_5, w_8, 0)}$ and the path fragments ($w_4$, $w_9$), ($w_5$, $w_8$) cannot be distinguished either. Ultimately, we can uniquely recover the query pairs $(w_{11}, v)$, $(w_2, v)$, $(w_{10}, v)$, $(w_3, v)$, $(w_7, v)$, and $(w_6, v)$. The query pairs ($w_9$, $v$) and ($w_8$, $v$) are mapped by the same token, while the token sets for ($w_1$, $v$), ($w_0$, $v$), ($w_4$, $v$), and ($w_5$, $v$) are indistinguishable, thus they cannot be directly recovered.

\begin{figure*}[!htbp]
  \centering
  \begin{minipage}{0.7\textwidth}
      \centering
      \includegraphics[width=0.98\linewidth]{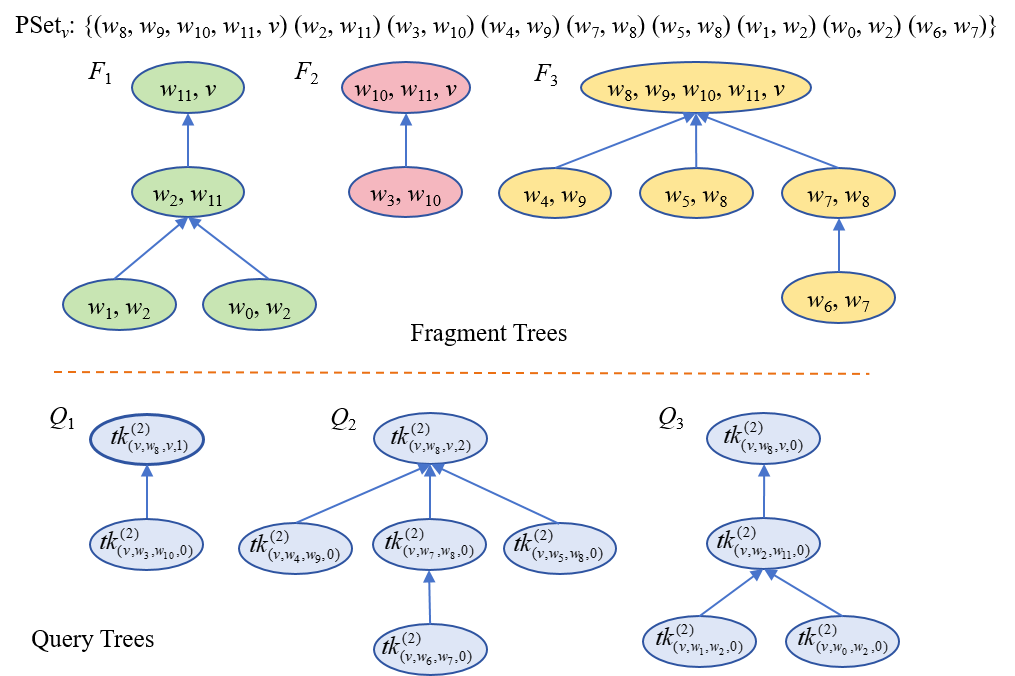}
      \caption{Examples of fragment trees and query trees constructed by our attack.}
      \label{FragmentTree}
    \end{minipage}
\end{figure*}

The above example illustrates that our query recovery attack method can effectively launch an attack on PathGES. The key factors enabling our attack method to achieve effective results include: What is a fragment tree? How to construct fragment trees based on $PSet$? How to build query trees based on $M_1$? How can we prove that the constructed fragment trees and query trees are isomorphic? Below, we will clarify these key issues one by one in three subsections.

\subsection{Fragment Tree Construction}
\subsubsection{Definition of Fragment Tree}
A fragment tree is a binary tuple $F$ $=$ ($P(F)$, $E(F)$), where:
\begin{itemize}
  \item $P(F)$ is a set of nodes, each corresponding to a collection of disjoint path fragments
  \item $E(F)$ a set of directed edges representing concatenation relationships between fragments
\end{itemize}

\noindent The fragment tree takes disjoint path fragments as nodes, with the path fragment containing the destination vertex serving as the root. Suppose there are two disjoint path fragments, $p_1$ $=$ ($w$, $\cdots$, $w'$) and $p_2$ $=$ ($u$, $\cdots$, $w$), and they share a connecting anchor $w$. Thus, $p_1$ is the direct predecessor of $p_2$ (i.e., the relationship $p_1$ $||$ $p_2$ exists). There is a directed edge from the node corresponding to $p_2$ to the node corresponding to $p_1$ in the corresponding fragment tree.

\subsubsection{Fragment Tree Construction Algorithm}
Given an SDSP tree $T_v$ (rooted at $v$) of the given graph $G$, the server firstly executes the HLD algorithm on $T_v$ and obtains the path fragment set $PSet_v$. Then, it takes set $PSet_v$ and the root $v$ as inputs and invokes Algorithm \ref{Alg_BuildFragTrees} to generate the corresponding fragment trees. Algorithm \ref{Alg_BuildFragTrees} consists of two phases: 1) preprocessing of path fragments (steps 2-10), and 2) queue-based tree construction (steps 11-27).

\begin{figure}[!t]
	\begin{algorithm}[H]
		\caption{Building Fragment Trees}\label{Alg_BuildFragTrees}
		\LinesNumbered
		\KwIn{$PSet_v$, the root $v$.}
		\KwOut{A set $\mathcal{F}_v$ of fragment trees.}
        Initialize $PS$ $=$ $\emptyset$; $\mathcal{F}_v$ $=$ $\emptyset$;\\
        \While{each path fragment $p$ of $PSet_v$}
        {
            \eIf{the length of $P$ is $1$}
            {
                $PS$ = $PS$ $\bigcup$ $\{p\}$;
            }{
                \If{$2^{j-1}$ $<$ the length of $P$ $\leq$ $2^j$}
                {
                    \For{$k = 0$; $k \leq j -1$; $k++$}
                    {
                        $p'$ $\leftarrow$ The sub-fragment of length $2^k$ at the rightmost end of $p$;\\
                        $PS$ = $PS$ $\bigcup$ $\{p'\}$;
                    }
                    $PS$ = $PS$ $\bigcup$ $\{p\}$;
                }
            }
        }
        \While{$PS$ is not empty}
        {
            \eIf{there exists a path fragment ending with $v$ in $PS$}
            {
                Initialize an empty fragment tree $F$ $=$ ($P(F)$, $E(F)$); \\
                Initialize an empty queue $Queue$; \\
                Remove this fragment from $PS$ and insert it into $Queue$; \\
                \While{$Queue$ is not empty}
                {
                    $p_2$ $\leftarrow$ Delete the front element of $Queue$;
                    Create a node $node_2$ based on $p_2$ and incorporate it into $P(F)$; \\
                    \If{$node_2$ is the first node to be created}
                    {
                        Use it as the root of $F$;
                    }
                    \While{there exists a path fragment $p_1$ in $PS$ such that $p_1$ $||$ $p_2$ holds}
                    {
                        Remove $p_1$ from $PS$ and insert it into $Queue$; \\
                        Create a node $node_1$ based on $p_1$ and incorporate it into $P(F)$; \\
                        Add an edge from $node_1$ to $node_2$ into $E(F)$; \\
                        $p_2 \leftarrow p_1$; $node_2 \leftarrow node_1$;
                    }
                }
                $\mathcal{F}_v$ $=$ $\mathcal{F}_v$ $\bigcup$ $\{F\}$;
            }{
                break;
            }
        }
        \textbf{Return} $\mathcal{F}_v$.
	\end{algorithm}
\end{figure}

Steps 2 to 10 of Algorithm \ref{Alg_BuildFragTrees} involve preprocessing $PSet_v$. This preprocessing converts a long path fragment into multiple sub-path fragments of different lengths that share the same ending vertex. For example, a path fragment of length 4 with the last vertex being $w$ will be converted into three sub-path fragments of lengths 1, 2, and 4 respectively, all ending with vertex $w$; a path fragment of length 6 with the last vertex being $w$ will be converted into four sub-path fragments of lengths 1, 2, 4, and 6 respectively, all ending with vertex $w$. For another example, the path fragment ($w_8$, $w_9$, $w_{10}$, $w_{11}$, $v$) in Fig. \ref{FragmentTree} is converted into three sub-path fragments ($w_{11}$, $v$), ($w_{10}$, $w_{11}$, $v$), and ($w_8$, $w_9$, $w_{10}$, $w_{11}$, $v$). Among the preprocessed path fragments, more than one fragment may include the root $v$. Therefore, one $PSet_v$ corresponds to a fragment-tree forest.

\subsection{Query Tree Construction}
\subsubsection{Background: Query Tree Definition}
As defined in Section \ref{datastructure}, a query tree $Q = (TK, E')$ is a tree constructed by the server after collecting all leaked tokens, where:
\begin{itemize}
  \item $TK$ is the set of all leaked tokens observed in the queries, with each token corresponding to a node
  \item $E'$ is a set of directed edges representing concatenation relationships between tokens
\end{itemize}

\noindent For a fixed destination vertex $v$, the query tree $Q_v$ has $tk_{(v,v)}$ as its root. If a token $tk_{w,v}$ is observed to be followed by $tk_{(v,v)}$ in a query sequence, a directed edge is added from $tk_{w,v}$ to the root.

\subsubsection{Token Sorting via Bucket-Based Method}
After successfully constructing fragment trees, we need to build query trees that are isomorphic to them. Since $EM_1$ adopts a response-revealing multi-map encryption algorithm, after the client completes all SDSP queries, the server can obtain full information about $M_1$. Our attack method use $M_1$ to construct query trees isomorphic to fragment trees. The key challenge is how to discover the direct predecessor relationships between the tokens stored in $M_1$. We have designed a bucket-based multi-round local sorting algorithm for $M_1$ (see Algorithm \ref{alg: TokenSorting}) to address this critical issue.

\begin{figure}[!t]
	\begin{algorithm}[H]
		\caption{Bucket-based Token Sorting}\label{alg: TokenSorting}
		\LinesNumbered
		\KwIn{$M_1$.}
		\KwOut{Sorted Buckets $D_1$, $\cdots$, $D_Z$.}
        Bucketing and deduplicating $M_1$ yields $D_1$, $\cdots$, $D_Z$; \\
        \For{$z = 2$; $z \leq Z$; $z++$}
        {
            \For{$j = 1$; $j \leq |D_{z-1}|$; $j++$}
            {
                \For {$k = 1$; $k \leq |D_z|$; $k++$}
                {
                    \If{$TK^z_{k}$ $\supset$ $TK^{z-1}_{j}$}
                    {
                        $TK^z_{k}$ $\leftarrow$ Merge the ordered sequences $TK^{z-1}_{j}$ and $TK^z_{k}$ $-$ $TK^{z-1}_{j}$;
                    }
                }
            }
        }
        \textbf{Return} $D_1$, $\cdots$, $D_Z$.
	\end{algorithm}
\end{figure}

\textbf{Example.} we take the $M_1$ in Fig. \ref{EM1_EM2} as an example to elaborate on the basic idea of this algorithm. The specific processing process is shown in Fig. \ref{order}.

First, we perform bucketing processing on $M_1$ based on the number of stored tokens. Storage locations with the same number of tokens belong to the same bucket, and duplicate entries within the bucket are removed (see Fig. \ref{order} (a) $\rightarrow$ (b)). We use $D_z$ to denote the set of storage locations that store $z$ tokens, and $TK^z_j$ to denote the set of tokens stored in the $j$-th storage location within $D_z$.

Next, we sort the tokens in each storage location within each bucket. When $M_1$ is divided into $Z$ buckets, a total of $Z-1$ rounds of part-range sorting will be performed. In the $z$-th ($z$ $=$ $2$, $\cdots$, $Z$) round of local sorting, the tokens stored in each location of $D_z$ are sorted based on the pre-sorted tokens in $D_{z-1}$. The so-called ordered bucket $D_{z-1}$ means that all $TK^{z-1}_1$, $\cdots$ $TK^{z-1}_{|D_{z-1}|}$ are ordered sequences. It is important to note that the tokens in each storage location of $D_1$ are already in a trivial order. The rule for sorting $D_z$ based on the pre-sorted $D_{z-1}$ is (see Fig. \ref{order} (b) $\rightarrow$ (c), (c) $\rightarrow$ (d)): For each storage location in $D_z$, identify the tokens that appear in $D_{z-1}$, arrange them in the order they appear in $D_{z-1}$ at the front, and place the remaining token in that location at the end. Obviously, the number of tokens stored in each location of $D_z$ is exactly one more than that stored in each location of $D_{z-1}$.

\begin{figure*}[!htbp]
  \centering
  \begin{minipage}{1.0\textwidth}
      \centering
      \includegraphics[width=0.98\linewidth]{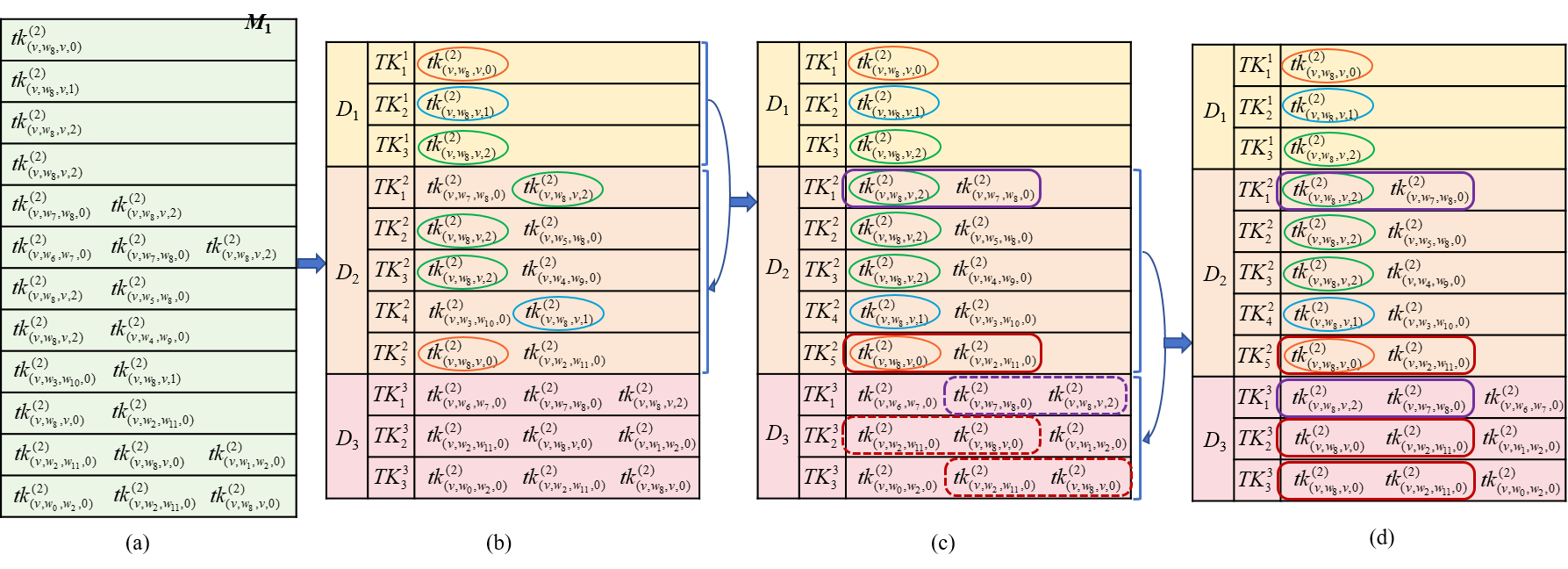}
      \caption{Illustrates the process of bucket-based token sorting. }
      \label{order}
    \end{minipage}
\end{figure*}

\subsubsection{Query Tree Generation Algorithm}
After obtaining the ordered $D_1$, $\cdots$, $D_Z$, the server generates query trees in a layer-by-layer manner using Algorithm \ref{Alg_BuildQueryTrees}). The nodes corresponding to the tokens in $D_1$ are the roots of query trees. Then, a node is generated for the $z$-th token in each token set within $D_z$, and a directed edge is subsequently added from this node to the node corresponding to its immediate predecessor token in the token set (i.e., the $(z-1)$-th token).

\begin{figure}[!t]
	\begin{algorithm}[H]
		\caption{Building Query Trees}\label{Alg_BuildQueryTrees}
		\LinesNumbered
		\KwIn{$M_1$.}
		\KwOut{Query Trees $Q_1$, $\cdots$, $Q_{|D_1|}$.}
        $D_1$, $\cdots$, $D_Z$ $\leftarrow$ Bucket-basedTokenSorting($M_1$); \\
        \For{$i = 1$; $i \leq |D_1|$; $i++$}
        {
            Create a node for the token in $TK^1_i$ to serve as the root of $Q_i$;
        }
        \For{$z = 2$; $z \leq Z$; $z++$}
        {
            \For{$j = 1$; $j \leq |D_z|$; $j++$}
            {
                Create a node for the $z$-th token in $TK^z_j$; \\
                Add a directed edge from this node to the node corresponding to the $(z-1)$-th token in $TK^z_j$;
            }
        }
        \textbf{Return} $Q_1$, $\cdots$, $Q_{|D_1|}$.
	\end{algorithm}
\end{figure}

\textbf{Example.} Taking the ordered buckets in Fig. \ref{order}(d) as an example, we will explain how to construct the query trees shown in Fig. \ref{FragmentTree}. The fact that there are three token sets in $D_1$ in Fig. \ref{order}(d) means that three query trees $Q_1$, $Q_2$, and $Q_3$ will be generated, with the tokens $tk_{(v, w_8, v, 1)}$, $tk_{(v, w_8, v, 2)}$, and $tk_{(v, w_8, v, 0)}$ in these three token sets serving as their respective roots. Because there are five token sets in $D_2$ in Fig. \ref{order}(d), there are a total of five nodes in the second layer, which correspond to the second tokens of these five token sets respectively. Three of them point to the root node corresponding to $tk_{(v, w_8, v, 2)}$, one of them points to the root node corresponding to $tk_{(v, w_8, v, 1)}$, and one points to the root node corresponding to $tk_{(v, w_8, v, 0)}$. By analogy, there are three nodes in the third layer, which correspond to the third tokens of the three token sets in $D_3$ respectively. Among them, two point to the second-layer node corresponding to $tk_{(v, w_2, w_{11}, 0)}$, and one points to the second-layer node corresponding to $tk_{(v, w_7, w_8, 0)}$.

\subsection{Isomorphism Proof}
This section establishes the theoretical foundation for the attack by proving the isomorphism between fragment trees and query trees.

\subsubsection{Definitions}
\textbf{Definition 5.1 (Consistency).} Let $G = (V, E)$ be a graph, and let the query sequence $q = (q_1, q_2, \dots, q_k)$ denote the $k$ shortest path queries issued by the client, where each query $q_i$ corresponds to a node pair $(u, v)$. Let $TK = (TK_1^1, TK_2^1, \dots, TK_k^x)$ denote the sequence of token sets leaked for all queries. If there exists an assignment $\pi : TK \to V \times V$ is a mapping from token set to SDSP queries,such that $L(G, q) = L(G, \pi(TK))$, then the mapping $\pi$ is said to be consistent with the leakage $L(G, q)$.

Informally, consistency requires that each query $(u, v)$ can correspond to an observable leaked token set.

\textbf{Definition 5.2 (Query Recovery).}  Given a query set $q = (q_1, q_2, \dots, q_k)$, where each $q_i$ denotes a shortest path query, and the corresponding leaked token sets are $TK = (TK_1^1, TK_2^1, \dots, TK_k^x)$, $\Pi$ is the set of all query-token mappings consistent with $L(G, q)$. If the adversary is able to output a mapping such that for any leaked token set $TK_m^n \in TK$, it can provide a set of possible original queries $\{\pi(TK_m^n) \mid \pi \in \Pi\}$, then we say the adversary achieves \emph{Query Recovery (QR)}.

Informally, query recovery requires that each $TK_m^n$ can be mapped through $\pi(TK_m^n)$ to determine a set of approximate queries consistent with the leakage. Although the approximate query set may not be unique, the correct query pair must be contained within each approximate query set. To explain our query recovery, we take an example from Fig. \ref{FragmentTree} and Fig. \ref{order}:

\textbf{Example.} In Fig.\ref{order}, we observe the token leaked $TK^1_1 = \{tk_{(v, w_8, v, 0)}^{(2)}\}$, $TK_5^2 = \{tk_{(v,w_8,v,0)}^{(2)},  tk_{(v, w_2,w_{11}, 0)}^{(2)}\}$, $TK_2^3 = \{tk_{(v,w_8,v,0)}^{(2)},  tk_{(v, w_2,w_{11}, 0)}^{(2)}, tk_{(v, w_1,  w_2 , 0)}^{(2)}\}$ and $TK_3^3 = \{tk_{(v,w_8,v,0)}^{(2)},  tk_{(v, w_2,w_{11}, 0)}^{(2)}, tk_{(v, w_0,  w_2 , 0)}^{(2)}\}$. Attackers can use these leaked tokens to construct the query tree $Q_3$ shown in Fig. \ref{FragmentTree}, which is isomorphic to $F_1$.To achieve query recovery, the attacker must produce the following outputs:
\[
\begin{aligned}
& \{\quad TK^1_1: \{(w_{11}, v)\} \quad TK^2_5: \{(w_2, v)\}  \\
& TK^3_2: \{(w_0, v), (w_1, v)\}\quad TK^3_3: \{(w_0, v), (w_1, v)\} \quad \}.
\end{aligned}
\]
Among the above four outputs, $TK^1_1$ and $TK^2_5$ can be accurately recovered, while $TK^3_2$ and $TK^3_3$ cannot uniquely determine the corresponding query pairs due to indistinguishable topology. It is worth noting that the failure to precisely recover a query pair may result not only from indistinguishable topological structures, but also from ambiguity caused by overly long path segments. For example, if the leaked token $tk_{(v, w_8, v, 2)}^{(2)}$ is mapped to the path segment $(v, w_{11}, w_{10}, w_9, w_8)$, we can determine that the destination is $v$, but it is difficult to distinguish whether the starting point is $w_8$ or $w_9$.

\textbf{Definition 5.4 (Fully Query Recovery).}
If the adversary has achieved query recovery (QR) and satisfies $|\Pi| = 1$, then the adversary is said to have achieved \emph{fully query recovery (FQR)}.

Informally, \emph{Fully Query Recovery (FQR)} refers to a precise form of query recovery, where the approximate query set recovered by the attacker contains only one unique query pair.

We further elaborate on the properties of isomorphic trees and prove that two trees with identical path name sets are isomorphic. In the Falzon-Paterson attack, consider two shortest path trees $T_r = (V, E)$ and $T'_{r'} = (V', E')$ rooted at different nodes $r$ and $r'$, respectively. Suppose there exists a path $P_{u, r} = (u, w_1, w_2, \dots, w_t, r)$ from node $u$ to root $r$ in $T_r$, and a path $P_{u', r'} = (u', w'_1, w'_2, \dots, w'_\ell, r')$ from node $u'$ to root $r'$ in $T'_{r'}$.

\subsubsection{Core Lemmas and Theorems}
\textbf{Theorem 5.1.} If there exists a tree isomorphism $\varphi: T_r \rightarrow T'_{r'}$ such that $\varphi(u) = u'$, then it follows that the path lengths must be equal (i.e., $t = \ell$), and for any $i \in [1, t]$, we have $\varphi(w_i) = w'_i$.

In other words, the paths strictly correspond one-to-one under the isomorphism. The Theorem 5.1 is proved in the Falzon-Paterson attack\cite{falzon2022efficient}.The fragment tree we construct is essentially a structural compression of the original SDSP tree. It aggregates consecutive path nodes in the original tree into segments, forming a coarser-grained hierarchical structure. The root of the fragment tree corresponds to the segment containing the root node $r$ and the fragment tree is logically equivalent to the SDSP tree. It still preserves the core properties of the SDSP tree: for example, the path from any node to the root is acyclic and unique within the fragment tree, and the path relationships between any two nodes remain unchanged.

Let $F = (P(F), E(F))$ and $F' = (P'(F'), E'(F'))$ be two two isomorphic  fragment trees. For $F$, there exists a path from a start node $u$ to an end node $r$, whose sequence of path fragments is denoted by \(P_{u,r} = (p_u,p_1, p_2, \dots, p_a,p_r)\), where each $p_i$ consists of a sequence of several nodes, with the start node \(u \in p_u\) and the end node \(r \in p_r\). Similarly, for $F'$, there exists a path from start node \(u'\) to end node \(r'\), whose path fragment sequence is \(P'_{u', r'} = (p'_{u'},p'_1, p'_2, \dots, p'_b,p'_{r'})\), where each $p'_i$ consists of a sequence of several nodes, with \(u' \in p'_{u'}\) and \(r' \in p'_{r'}\).

\textbf{Lemma 5.1.} If there exists a fragment-to-fragment mapping \(\varphi: F \to F'\) such that \(\varphi(p_u) = p'_{u'}\), then it follows that \(a = b\), and for all \(i \in [1,a]\), we have \(\varphi(p_i) = p'_i\).

\textbf{Proof.} Given that $\varphi(p_u) = p'_{u'}$ and according to the definition of rooted tree isomorphism, we also have $\varphi(p_r) = p_{r'}$. The isomorphism between graphs preserves edges; therefore, the mapping function $\varphi$ must map the path fragment sequence $P_{u,r}$ to $P'_{u',r'}$. In other words, the number of path fragments forming these two shortest paths must be equal, hence $a = b$. Combining the above observations, we obtain the correspondence between the edges of the two isomorphic trees as follows:\\
$(\varphi(p_u), \varphi(p_1)) = (p'_{u'}, p'_1)$, $(\varphi(p_1), \varphi(p_2)) = (p'_1, p'_2),\dots,
 \\(\varphi(p_a), \varphi(p_r)) = (p'_b, p'_{r'})$. \\Thus, it is proven that \(\varphi(p_i) = p'_i\).
\hfill $\square$

Let $Q = (TK, E')$ be a query tree, and the leaked and sorted token set for the query $(u', r')$ be $TK_{(u', r')} = (tk_{u'},tk_1, tk_2, \dots, tk_c,tk_{r'})$. To facilitate the description of the logical and quantitative relationship between the token $tk_{(r,u,v,j)}$ and the path fragment $p$, we abbreviate $tk_{(r,u,v,j)}$ as $tk_i$. The path fragment mapped by \(tk_{u'}\) contains the start node \(u'\), and the path fragment mapped by \(tk_{r'}\) contains the root \(r'\).

Based on Lemma 5.1, we can derive a similar conclusion as follows:

\textbf{Lemma 5.2.} If there exists a mapping from fragments to tokens $\varphi: F \to Q$ suchs that $\varphi(p_u) = tk_{u'}$, then it must hold that $a = c$, and for any $i \in [1, a]$, we have $\varphi(p_i) = \mathsf{tk}_i$.

We use $F[p_i]$ to denote the subtree rooted at the path fragment $p_i$, and $Q[tk_i]$ to denote the subtree rooted at the token $tk_i$.

\textbf{Theorem 5.3.}
Let $\varphi$ be an isomorphic mapping from $F$ to $Q $. This holds if and only if there exists a perfect matching between each $p_i \in F$ and $tk_i \in Q$ such that any subtree $F[p_i]$ in $F$ can be mapped to the subtree $Q[tk_i]$.

\textbf{Proof.}
Suppose $\varphi$ is an isomorphic mapping from $F$ to $Q$, and $(p_i, p_{i+1})$ is an edge in the fragment tree. If $\varphi(p_i) = tk_i$ and $\varphi(p_{i+1}) = tk_{i+1}$ with $p_i, p_{i+1} \in F$, then by the edge-preserving property of isomorphism, $\varphi$ must ensure the existence of the edge $(tk_i, tk_{i+1})$, i.e., $\varphi(p_i, p_{i+1}) = (tk_i, tk_{i+1})$.
Extending this property to subtrees yields $F[p_i] \cong Q[tk_i]$.
\hfill $\square$

\subsubsection{PathName-Based Isomorphism}
The computation of path names is the key to finding an isomorphic fragment tree for the query tree $Q$. Recall from Section \ref{TreeIsoCanNames} taht PathName encodes the path structure from a node to the root. We denote by $PathName_F(p_i)$ the pathname from a node $p_i$ to the root node $p_r$ in the fragment tree $F$. According to the definition of pathnames, we have:\\
$PathName_F(p_u) = h(Name(p_u)  \Vert  PathName_F(p_1))$,\\
$PathName_F(p_1) = h(Name(p_1)  \Vert  PathName_F(p_2))$,\\
$\dots$\\
$PathName_F(p_c) = h(Name(p_c)  \Vert  PathName_F(p_r))$,\\
$PathName_F(p_r) = h(Name(p_r))$\\
Similarly, we denote by $PathName_Q(tk_{i})$ the pathname from a token $tk_{i}$ to the root token $tk_{r’}$ in the query tree $Q$. According to the definition of pathnames, we have:\\
$PathName_Q(tk_{u'}) = h(Name(tk_{u'}) \Vert PathName_Q(tk_1))$,\\
$PathName_Q(tk_1) = h(Name(tk_1) \Vert PathName_Q(tk_2))$,\\
$\dots$\\
$PathName_Q(tk_{c}) = h(Name(tk_{c}) \Vert PathName_Q(tk_{r'}))$,\\
$PathName_Q(tk_{r'}) = h(Name(tk_{r'}))$

\textbf{Theorem 5.2.} If the fragment tree $F = (P(F), E(F))$ and the query tree $Q = (TK, E')$ are isomorphic, and the mapping function $\varphi$ maps $p_u \in P_{u,v}$ to $tk_{u'} \in TK_{u',r'}$, then it holds that $PathName_{F}(p_u) = PathName_Q(tk_{u'})$.
Conversely, if $PathName_{F}(p_u) = PathName_{Q}(tk_{u'})$, then $F \cong Q$, and there exists a mapping function $\varphi$ that maps $p_i$ to $tk_i$.

\textbf{Proof.} For the forward direction, since $F_{pr}$ and $Q_{r'} $ are isomorphic, by Theorem 4.9 we know that $F[p_i] \cong Q[tk_i]$. For any $p_i$ and $tk_i$, we have $Name(p_i) = Name(tk_i)$.
Moreover, since $p_u$ is mapped to $tk_{u'}$, by Lemma 4.7 we have $a = c$.
Finally, by the definition of path names, the path name is the concatenation of the canonical names from the root to the start node in each tree. Hence, we conclude that $PathName_F(p_u) = PathName_{Q}(tk_{u'})$.

For the reverse direction, from $PathName_F(p_u) = PathName_Q(tk_{u'})$ we obtain $Name(p_u) = Name(tk_{u'})$, $Name(p_1) = Name(tk_1)$, $\cdots$, $Name(p_r) = Name(tk_{r'})$. Thus, $F[p_i] \cong Q[tk_i]$, and therefore $F \cong Q$. Hence, there exists a mapping function $\varphi$ that perfectly matches $p_i$ and $tk_i$.
\hfill $\square$

\textbf{Corollary 5.1.} According to Theorem 4.8, if the fragment tree $F$ and the query tree $Q$ are isomorphic, and each path fragment $p_i$ in the fragment tree and each token $tk_i$ in the query tree $Q $ has a unique path name, then intuitively, every $p_i$ in $F$ can only be mapped to the token $tk_i$ in $Q$ that has the same path name. Similarly, each token $tk_i$ will also be mapped to the path fragment $p_i$ with the same path name.

\subsection{Attack Effectiveness}
We conducted the attack experiments on each dataset under the same environment using Python 3.10. The experimental server is equipped with six NVIDIA 8000-48G GPUs, each with 48GB VRAM. The GPU driver version is 535.104.12, supporting up to CUDA version 12.2.
The CPU is an Intel Xeon E5-2699 v4 with 11 cores and 62.9GB memory. NetworkX (version 3.5) was used for graph-related operations, and the shortest paths were computed via a single-source shortest path algorithm. NumPy (version 2.3.1) was used for numerical computations. Matplotlib (version 3.10.3) was used for visualization. For cryptographic primitives, we used the Cryptography library (version 45.0.5). Symmetric encryption employed AES-CBC mode (with 16-byte keys), and SHA-256 was used as the hash function. Tokens were generated using HMAC-SHA256, truncated to 128 bits.
The source code used to implement and evaluate the Fragment Tree attack is publicly available at: https://anonymous.4open.science/r/Fragment-Trees/

\subsubsection{Graph Datasets}\label{DataSets}
We conducted attack experiments on seven datasets, which are the same datasets used in the Falzon–Paterson attack. All datasets are derived from real-world networks, with the number of vertices ranging from as few as 35 to as many as 22,687. The graph densities range from 0.0002 to 0.995. Below, we provide a brief description of these seven datasets.

\textbf{InternetRouting}: This dataset originates from the University of Oregon Route Views Project and was extracted using the dense subgraph extraction algorithm proposed by Charikar\cite{charikar2000greedy}, as adopted by Ambavi et al.

\textbf{Ca-GrQc}: This dataset represents the collaboration network in the field of General Relativity and Quantum Cosmology on arXiv.org, covering the period from January 1993 to April 2003. It was also extracted as a subgraph using the dense subgraph extraction algorithm.

\textbf{email-Eu-core}: This dataset captures email communications among researchers. An edge exists between two nodes if one sent an email to the other or received an email from the other.

\textbf{facebook}: This dataset is derived from Facebook’s friendship network. An edge exists between two users if they are friends.

\textbf{p2p-Gnutella08, p2p-Gnutella04, p2p-Gnutella25}: These three datasets describe snapshots of the Gnutella peer-to-peer (P2P) network collected on August 4 and August 8, and August 25, 2002, respectively. Nodes represent hosts within the Gnutella network topology, and edges represent connections between these hosts.

\noindent Note that PathGES also uses the six datasets listed above, except for p2p-Gnutella25.The datasets used in our evaluation are publicly available and
can be accessed at:http://snap.stanford.edu/data.

\subsubsection{Experimental Design}
We conducted experiments in three aspects: 1) We executed the HLD algorithm on seven datasets to calculate the proportions of path fragments with lengths 1, 2, and 3, aiming to verify whether PathGES generates a large number of one-to-one mapping relationships during actual execution. 2) In the dimension of accurate query recovery attacks, we compared the attack effectiveness of the Falzon-Paterson attack on GKT and that of our attack method on PathGES (an enhanced version of GKT) based on seven datasets, thereby illustrating that our method can launch effective query recovery attacks against PathGES, which is resistant to the Falzon-Paterson attack. 3) Based on six datasets, we evaluated the effectiveness of our method in approximate query recovery when launching query recovery attacks against PathGES.

The experimental results of the first aspect are shown in Table \ref{tab:datasets2}. These results indicate that PathGES indeed generates a large number of one-to-one mapping relationships when executed on real-world datasets, verifying the deficiency of PathGES in hiding topological structure information that we identified based on its implementation principle, and laying the foundation for the effective implementation of the attack method we proposed.

\begin{table*}[ht]
    \centering
    \caption{Statistics on the Number of Path Fragments with Different Lengths}
    \label{tab:datasets2}
    \begin{threeparttable}
        \begin{tabular}{c|c|c|c|c|c|c}
            \hline
            Data & $|V|$ & $|E|$ & d & Total Fragments & Short Fragments & Percentage \\
            \hline
            internetRouting & 35 & 323 & 0.543 & 1141 & 1141 & 100\% \\
            Ca-GrQc & 46 & 1030 & 0.995 & 2069 & 2069 & 100\% \\
            email-EU-core & 1005 & 16,706 & 0.02 & 878147 & 877940 & 99.98\% \\
            facebook-combined & 4309 & 88234 & 0.011 & 16281762 & 16281560 & 99.99\% \\
            p2p-Gnutella08 & 6301 & 20777 & 0.001 & 34619060 & 34578445 & 99.88\% \\
            p2p-Gnutella04 & 10876 & 39994 & 0.0006 & 101440282 & 101354239 & 99.92\% \\
            p2p-Gnutella25 & 22687 & 54705 & 0.0002 & 472236309 & 472094187 & 99.97\% \\
            \hline
        \end{tabular}
        \begin{tablenotes}
            \footnotesize
            \item[1] $d = \frac{2|E|}{|V|^2 - |V|}$ represents the density of the graph.
            \item[2] Short fragments refer to those with lengths of 1, 2, and 3.
            \item[3] Percentage indicates the proportion of the short fragments (lengths 1, 2, or 3) among all fragments.
        \end{tablenotes}
    \end{threeparttable}
\end{table*}

The experimental results of the second aspect are shown in Table \ref{tab:Result_Attack}. The results show that, for the Falzon-Paterson attack and our attack method, the worst attack performance occurred on the Ca-GrQc graph, mainly due to its extremely high edge density of 0.995. In other words, most nodes in Ca-GrQc are almost directly connected, and the shortest path between any two nodes is often the direct edge connecting them. This results in a highly homogeneous topology for shortest paths, with path names being very similar across different paths, which severely weakens the distinguishability based on path names. The dataset with the best attack performance is p2p-Gnutella04, where approximately 10.236\% of queries can be uniquely matched to the corresponding subpath fragments through leaked information. The effectiveness of our attack in achieving accurate query recovery against PathGES—which is equipped with defense mechanisms against such attacks—reaches approximately 50\% of that achieved by the Falzon–Paterson attack against GKT, which lacks such defenses, across four datasets. On two datasets, our attack attains about 30\% of the effectiveness of the Falzon–Paterson attack, while on one dataset, it reaches roughly 10\%. The experimental results show that our method can exploit PathGES's defense mechanism's weakness in hiding topological structure information to launch effective query recovery attacks against it.

\begin{table*}[ht]
    \centering
    \caption{Comparison of Accurate Query Recovery Attack Effects}
    \label{tab:Result_Attack}
    \begin{threeparttable}
        \begin{tabular}{c|c|cc|cc}
        \hline
        \multirow{2}{*}{\textbf{Data}} & \multirow{2}{*}{\textbf{Total Query}} & \multicolumn{2}{c|}{\textbf{Falzon-Paterson to GKT}} & \multicolumn{2}{c}{\textbf{Our Attack to PathGES}} \\
        \cline{3-6}
        & & \textbf{Unique} & \textbf{Percentage} & \textbf{Unique} & \textbf{Percentage} \\
        \hline
        internetRouting    & 1190       & 28       & 2.353\%  & 5         & 0.420\% \\
        Ca-GrQc            & 2070       & 3        & 0.145\%  & 1         & 0.048\% \\
        email-EU-core      & 1009020    & 65659    & 6.507\%  & 20645     & 2.046\% \\
        facebook-combined  & 16309482   & 33634    & 0.206\%  & 19168     & 0.118\% \\
        p2p-Gnutella08     & 39696300   & 8519868  & 21.463\% & 3771341   & 9.500\% \\
        p2p-Gnutella04     & 118276500  & 25915785 & 21.911\% & 12106435  & 10.236\% \\
        p2p-Gnutella25     & 514677282  & 82736533 & 16.075\% & 45327736  & 8.807\% \\
        \hline
        \end{tabular}
        \begin{tablenotes}
            \footnotesize
            \item[1] \textit{Unique} indicates the number of query that can be precisely recovered during the attack.
            \item[2] \textit{Percentage} represents the proportion of such queries among all queries.
        \end{tablenotes}
    \end{threeparttable}
\end{table*}

The approximate query recovery attack is an attack where the target query cannot be accurately recovered; however, it can be confined to a limited range. To characterize this attack effect, we introduce the concept of \emph{query recovery capability}. The query recovery capability for the accurate query recovery analyzed earlier is regarded as 100\%. Suppose an attack can narrow the recovered target query down to two indistinguishable approximate queries, one of which is the target query. In that case, the query recovery capability of this attack is 50\%. Similarly, a ``33\% query recovery capability" means that the attacker narrows the target query down to three indistinguishable approximate queries. ``25\%" and ``20\%" query recovery capabilities correspond to scenarios where the target query range is narrowed down to four or five indistinguishable approximate queries, respectively. In other words, a query recovery attack with a query recovery capability of 100\% is an accurate query recovery attack; a query recovery attack with a query recovery capability less than 100\% is an approximate query recovery attack. The higher the query recovery capability, the higher the query recovery accuracy.

The experimental results of the third aspect are shown as Fig. \ref{fig:Approximate_Attack}. The $y$-axis represents the proportion of queries that are successfully approximately recovered relative to the total number of queries. Fig. \ref{fig:Approximate_Attack} shows that as the graph becomes sparser, the proportion of successful approximate queries with high recovery capability increases. Although these attacks cannot accurately recover the query, they confine the target query to a small range, posing a threat to PathGES. For example, in real-world scenarios, attackers often have access to additional side information, such as knowledge of the source vertex or commonly queried destinations. By leveraging such information, attackers can further reduce the approximate query set, thereby turning approximate recovery into exact recovery and exacerbating privacy risks.

\begin{figure*}[!htbp]
    \centering
    \begin{minipage}{1.0\textwidth}
    \includegraphics[width=0.9\linewidth]{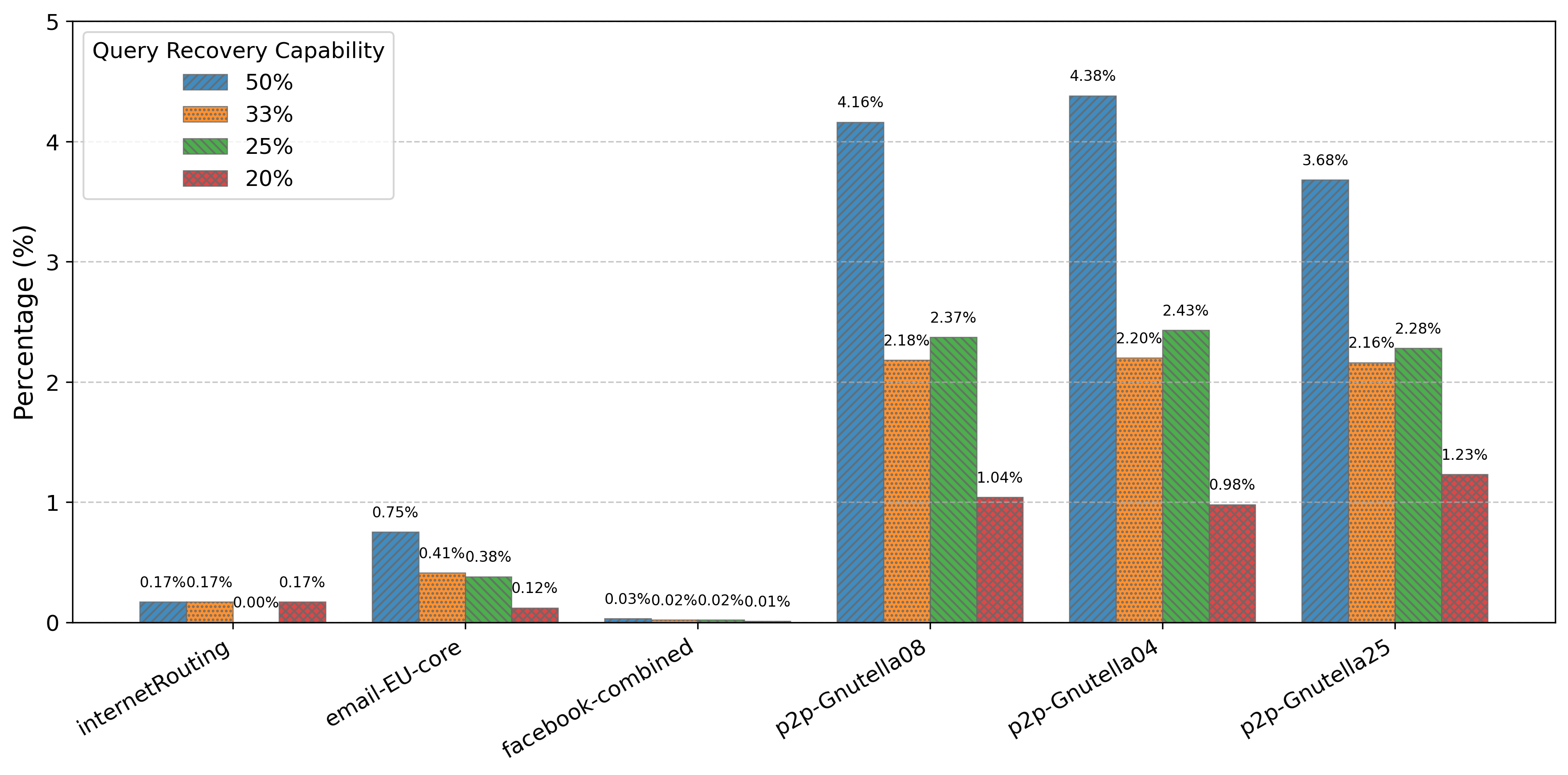}
    \caption{The Situation of Successfully Launching Approximate Attacks on PathGES}
    \label{fig:Approximate_Attack}
    \end{minipage}
\end{figure*}

It should be noted that the graph represented by the Ca-GrQc dataset has an extremely high density. Therefore, the target query cannot be confined to a small range, making it impossible to mount effective approximate query recovery attacks. For this reason, the experiments of the third aspect were conducted on six datasets, excluding the Ca-GrQc dataset.

\section{BlindGES Defense Scheme}
\subsection{Design Goals and Workflow}
BlindGES is designed to accurately retrieve the shortest path between two nodes while effectively resisting the Falzon-Paterson attack. Moreover, it ensures that the adversary cannot infer the length of the shortest path from the query responses.

In BlindGES, we construct a two-level index structure (see Algorithm \ref{Alg_Build_Index}) based on a Merge-and-Divide mechanism (see Algorithm \ref{Alg_MergeDivide}) that we introduced and the outputs $Pset$s of the HLD algorithm for a given graph $G$. In addition, we design a corresponding algorithm (see Algorithm \ref{Alg_RecoverPath}) to recover the final query results from the search results obtained based on the proposed two-level index structure. For ease of description, we take an SDSP tree $T_r$ of $G$ as an example below to elaborate on the core idea of BlindGES in constructing an index structure aimed at eliminating one-to-one mappings. As shown in Table\ref{tab:datasets2}, the proportion of one-to-one mappings in BlindGES is reduced by at least $80\%$ compared with PathGES.

\subsection{Merge-and-Divide Mechanism}
The merging target of the Merge-and-Divide mechanism is the paths of length 1 in $Pset_r$, and path concatenation is performed by introducing a dummy vertex `$q$'. The Merge-and-Divide mechanism adopts different division parameters $l_1$ and $l_2$ for merged paths and non-merged paths. Specifically, merged paths are divided into segments of length $l_1$, with the last segment padded with a dummy vertex `$p$' if necessary; each non-merged path is divided into segments of length $l_2$, and the last segment is also padded with the dummy vertex `$p$' if needed. The two-level index structure constructed by BlindGES consists of two mapping tables, $M_1$ and $M_2$. The decomposed segments are stored in $M_2$. Therefore, BlindGES assigns a corresponding label to each segment. Suppose a segment is the $j$-th segment obtained by BlindGES decomposing a certain path, and assume that the first non-$q$ vertex of this path is $u$ and the last non-$p$ vertex is $v$. Then, the label of this segment is $lab^{(2)}(r, u, v, j)$.

\begin{algorithm}
    \caption{Merge ($AL_{Mer}$).}
    \label{Alg_Merge}
    \LinesNumbered
    \KwIn{$Pset_r$.}
    \KwOut{The set of merged pairs $Pset'_r$.}

    Initialize $Pset'_r \gets \emptyset$; $buffer \gets \emptyset$;\\
    \For {$path \in Pset_r$}
    {
        \If{\text{the length of} $path$ \text{is one}}
        {
            add ``$q$'' $||$ $path$ into $buffer$;
        }
        \ElseIf{ $buffer \neq \emptyset$}
        {
            \If{$buffer$ has more than one path}
            {
                Initialize $merged\_path \gets \emptyset$;\\
                \For{$bpath \in buffer$}
                {
                    $merged\_path \gets merged\_path || bpath$;
                }
                add $merged\_path$ into $Pset'_r$;
                clear $buffer$;
            }\Else
            {
                $single\_path \gets buffer[0]$;\\
                delete $q$ from $single\_path$;\\
                add $single\_path$ into $Pset'_r$;
                clear $buffer$;
            }
        }
        add $path$ into $Pset'_r$;
    }
    \If{$buffer \neq \emptyset$}
    {
        \If{$buffer$ \text{has more than one path}}
        {
            \For{$bpath \in buffer$}
            {
                $merged\_path \gets merged\_path || bpath$;
            }
            add $merged\_path$ into $Pset'_r$;
        }
        \Else{
            $single\_path \gets buffer[0]$;\\
            delete $q$ from $single\_path$;\\
            add $single\_path$ into $Pset'_r$;
        }
    }
    \Return $Pset'_r$;
\end{algorithm}

\begin{algorithm}
    \caption{Divide ($AL_{Div}$).}
    \label{Alg_Divide}
    \LinesNumbered
    \KwIn{$Pset'_r$, division parameters $l_1$ and $l_2$.}
    \KwOut{The set of segment label pairs $GPset'_r$.}

    Initialize $GPset'_r \gets \emptyset$;\\
    \For {$path' \in Pset'_r$}
    {
        $num\_edges \gets |path'| - 1$;\\
        \If{`$q$' \text{is the first vertex of} $path'$}
        {
            $u$ $\gets$ the first non-$q$ vertex of $path'$; \\
            $v$ $\gets$ the last non-$q$ vertex of $path'$; \\
            $num \gets \lfloor num\_edges / l_1 \rfloor$;\\
            \For{$j = 0$ \KwTo $num - 1$}
            {
                $segment$ $\gets$ the subsequence of $path'$, from the $(j\times l_1)$-th vertex to the $((j+)\times l_1-1)$-th vertex;\\
                add ($segment$, $lab^{(2)}_{(r, u, v, j)}$) into $GPset'_r$;
            }
            $segment$ $\gets$ the subsequence of length $|path'|- num \times l_1$ at the end of $path'$;\\
            Pad $segment$ with `$p$'  to form a vertex sequence of length $l_1$;\\
            add ($segment$, $lab^{(2)}_{(r, u, v, num)}$) into $GPset'_r$;
        }
        \Else{
            $u$ $\gets$ the first vertex of $path'$; \\
            $v$ $\gets$ the last vertex of $path'$; \\
            $num \gets \lfloor num\_edges / l_2 \rfloor$;\\
            \For{$j = 0$ \KwTo $num - 1$}
            {
                $segment \gets$ the subsequence of $path'$, from the $(j\times l_2)$-th vertex to the $((j+)\times l_2-1)$-th vertex;\\
                add ($segment$, $lab^{(2)}_{(r, u, v, j)}$) into $GPset'_r$;
            }
            $segment$ $\gets$ the subsequence of length $|path'|- num \times l_2$ at the end of $path'$;\\
            Pad $segment$ with `$p$'  to form a vertex sequence of length $l_2$;\\
            add ($segment$, $lab^{(2)}_{(r, u, v, num)}$) into $GPset'_r$;
        }
    }
    \Return $GPset'_r$;
\end{algorithm}

\begin{algorithm}
    \caption{Merge and Divide ($AL_{MerDiv}$).}
    \label{Alg_MergeDivide}
    \LinesNumbered
    \KwIn{$Pset_r$, division parameters $l_1$ and $l_2$.}
    \KwOut{The set of segment label pairs $GPset'_r$.}

    Initialize $Pset'_r \gets \emptyset$; $GPset'_r \gets \emptyset$;\\
    $Pset'_r \gets AL_{Mer}(Pset_r)$;\\
    $GPset'_r \gets AL_{Div}(Pset'_r, l_1, l_2)$;\\
    \Return $GPset'_r$;
\end{algorithm}

\textbf{Example.} As shown in Fig.\ref{Exep_MergeSplit}, when merging the paths ($v_1$, $v_2$), ($v_1$, $v_3$), ($v_5$, $v_6$), and ($v_8$, $v_9$), we first insert the virtual vertex `$q$' at the beginning of each of them and concatenate them to obtain a single path ($q$, $v_1$, $v_2$, $q$, $v_1$, $v_3$, $q$, $v_5$, $v_6$, $q$, $v_8$, $v_9$). After executing $AL_{Mer}$ (Algorithm \ref{Alg_Merge}), $Pset'_r$ contains two types of paths: the merged paths that contain $q$ (such as the lower path in Fig.\ref{Exep_MergeSplit}) and the non-merged paths (such as the upper path ($u$, $w_6$, $w_5$, $w_4$, $w_3$, $w_2$, $w_1$, $v$) in Fig.\ref{Exep_MergeSplit}). The client needs to set two parameter $l_1$ and $l_2$ to define the length of segment. $l_1$ is an integer multiple of $3$ and $1$ $<$ $l_2$ $<$ $l_1$. Note that, in this paper, we represent a path or segment using a vertex sequence. Therefore, $l_1$ and $l_2$ are number of vertices of a segment. The length of segment is $l_1$ $-$ $1$ or $l_2$ $-$ $1$. The client needs to pad the dummy vertex `$p$' into merged (resp. non-merged) paths to make their length to be an integer multiple of $l_1$ (resp. $l_2$). A critical point is that when splitting a non-merging path, to prevent loss of path information, the starting vertex of the latter segment must overlap with the ending node of the previous segment. As shown in Fig.\ref{Exep_MergeSplit}, the path ($p$, $p$, $u$, $w_6$, $w_5$, $w_4$, $w_3$, $w_2$, $w_1$, $v$) should be split into ($p$, $p$, $u$, ${w_6}$), ($w_6$, $w_5$, $w_4$, $w_3$), and ($w_3$, $w_2$, $w_1$, $v$).

\begin{figure}[htbp]
  \centering
  \begin{minipage}{0.7\textwidth}
      \centering
      \includegraphics[width=\linewidth]{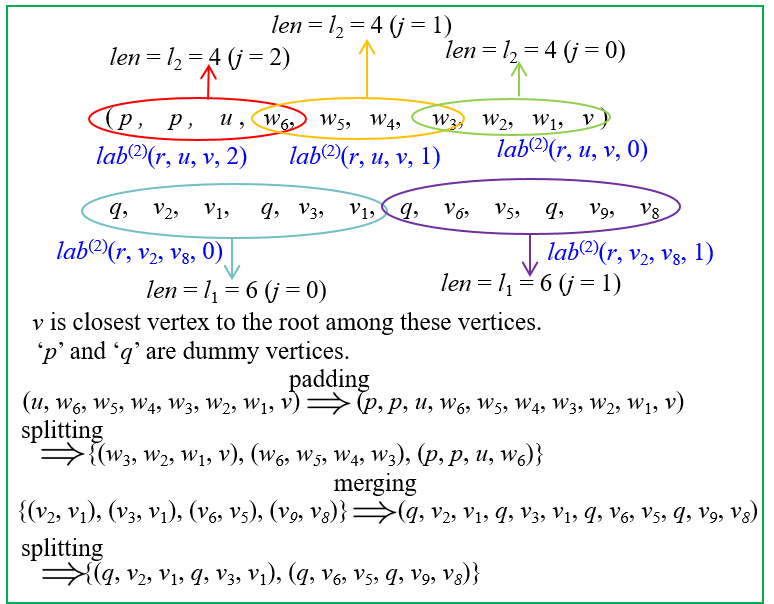}
      \caption{Example for the Merge-and-Divide mechanism.}
      \label{Exep_MergeSplit}
    \end{minipage}
\end{figure}

\subsection{Two-Level Index Structure}
In BlindGES, the two-level index structure consists of $M_1$ and $M_2$, both encrypted to enable secure retrieval. $M_2$, encrypted via the EMM-RH algorithm, is accessed using tokens $tk_{(r, u, v, j)}$ stored within $M_1$. $M_1$ itself, encrypted with EMM-RR, uses query tokens $tk_{(u, v)}$ corresponding to vertex pairs $(u, v)$ to retrieve all relevant canonical fragment tokens associated with the shortest path from $u$ to $v$. Thus, an SPSP query token $tk_{(u, v)}$ facilitates the recovery of all path segments constituting the shortest path between $u$ and $v$ through this encrypted two-level index.

\begin{algorithm}
    \caption{Index Structure Construction ($AL_{Index}$).}
    \label{Alg_Build_Index}
    \LinesNumbered
    \SetKwInput{KwIn}{Input}
    \SetKwInput{KwOut}{Output}
    \KwIn{$\{Pset_r\}_{r\in V(G)}$, division parameters $l_1$ and $l_2$, and $K_2$.}
    \KwOut{($M_1$, $M_2$).}

    Initialize $M_1 \gets \emptyset$, $M_2 \gets \emptyset$;\\
    \For{each $Pset_r$ $\in$ $\{Pset_r\}_{r\in V(G)}$}
    {
        $GPset'_r$ $\gets$ $AL_{MerDiv}(Pset_r, l_1, l_2$);\\
        \For{\text{each} $(segment, lab^{(2)}) \in GPset'$}
        {
            $M_2[lab^{(2)}] \gets segment$\;
            $tk \gets \text{EMM-RH.Token}(K_2, lab^{(2)})$\;
            $w$ $\gets$ $segment[0]$ (the first vertex of $segment$)\;
            \If{$w$ \text{is not} `$q$'}
            {
                $w$ $\gets$ the first non-$p$ vertex of $segment$\;
                $v'$ $\gets$ the third item of $lab^{(2)}$\;
                \While{$w$ \text{is not the last vertex of} $segment$}
                {
                    $M_1[(w, r)] \gets M_1[(v', r)] \cup [tk]$\;
                    \text{shuffle}($M_1[(w, r)]$)\;
                    $w$ $\gets$ the next vertex of $segment$\;
                }
            }
            \Else
            {
                $i \gets 1$\;
                \While{$i < l_1$}
                {
                    $w$ $\gets$ the $i$-th vertex of $segment$\;
                    $v'$ $\gets$ the $(i+1)$-th vertex of $segment$\;
                    \If{$w$ \text{is} `$p$' \text{or} $v'$ \text{is} `$p$'}
                    {
                        break\;
                    }
                    $M_1[(w, r)] \gets M_1[(v', r)] \cup [tk]$\;
                    \text{shuffle}($M_1[(w, r)]$)\;
                    $i$ $\gets$ $i+3$\;
                }
            }
        }
    }
    \Return $(M_1, M_2)$;
\end{algorithm}

\begin{algorithm}
	\caption{Shortest Path Recovery ($AL_{Reco}$). }
	\label{Alg_RecoverPath}
	\SetKwInput{KwIn}{Input}
	\SetKwInput{KwOut}{Output}
	\KwIn{A set of fragments $SegSet$, the source vertex $u$, the destination vertex $v$.}
	\KwOut{the shortest path $path_{(u,v)}$ from $u$ to $v$.}
	
	Initialize $Edges \gets \emptyset$\;
	Initialize $path_{(u,v)}$ as a vertex sequence containing the vertex $u$\;
	\For{\text{each} $segment \in SegSet$}
	{
		\For{$i$ \text{from} $0$ \text{to} $|fragment$|}
		{
			$a \gets fragment[i]$\;
			\If{$a$ \text{is} `$q$'}
			{
				continue\;
			}
			\ElseIf{$a$ \text{is} `$p$'}
			{
				break\;
			}
			\Else{
				$b \gets fragment[i+1]$\;
				\If{$b$ \text{is} `$p$'}
				{
					break\;
				}
				\Else{
					add $(a,b)$ into $Edges$\;
				}
			}
			$i$ $\gets$ $i+1$;
		}
	}
	Set $current \gets u$\;
	\While{$current \neq v$}{
		\While{$(a, b) \in Edges$ is an unvisited edge}
		{
			\If{$a == current$}
			{
				$current \gets b$\;
				add $current$ into $path_{(u,v)}$\;
				break\;
			}
		}
	}
	\Return{$path_{(u,v)}$}\;
\end{algorithm}

Note that, different from PathGES, the segments stored in each entry of $M_1$ are not necessarily continuous segments, which means these segments, when pieced together, do not form a path, effectively avoiding the problem of excessive information leakage caused by returning redundant paths in PathGES. These segments contain the ones that the client truly needs. Therefore, the client can extract the actual required path from the returned segments.

\subsection{Query Processing}
The construction of BlindGES is illustrated in detail in Fig. \ref{BlindGESImp}. In general, BlindGES consists of two phases: the setup phase and the query phase. The setup phase corresponds to the \texttt{BlindGES.KeyGen} algorithm and \texttt{BlindGES.Encrypt} algorithm. The client first generates various parameters by executing \texttt{BlindGES.KeyGen}. Then, the client constructs an encrypted two-level multimap index structure (as shown in Fig. \ref{setup_phase}) and sends it to the server. The query phase includes the \texttt{BlindGES.Token} algorithm, the \texttt{BlindGES.Search} algorithm, and the \texttt{BlindGES.Reveal} algorithm. The client generates a SPSP query by carrying out \texttt{BlindGES.Token} and sends the generated token to the server. The server executes \texttt{BlindGES.Search} based on the received token and the encrypted index structure to obtain the query result and returns it to the client. The client recover the corresponding shortest path from the received query result by executing \texttt{BlindGES.Reveal}.

\begin{figure*}[htbp]
  \centering
  \includegraphics[width=0.9\linewidth]{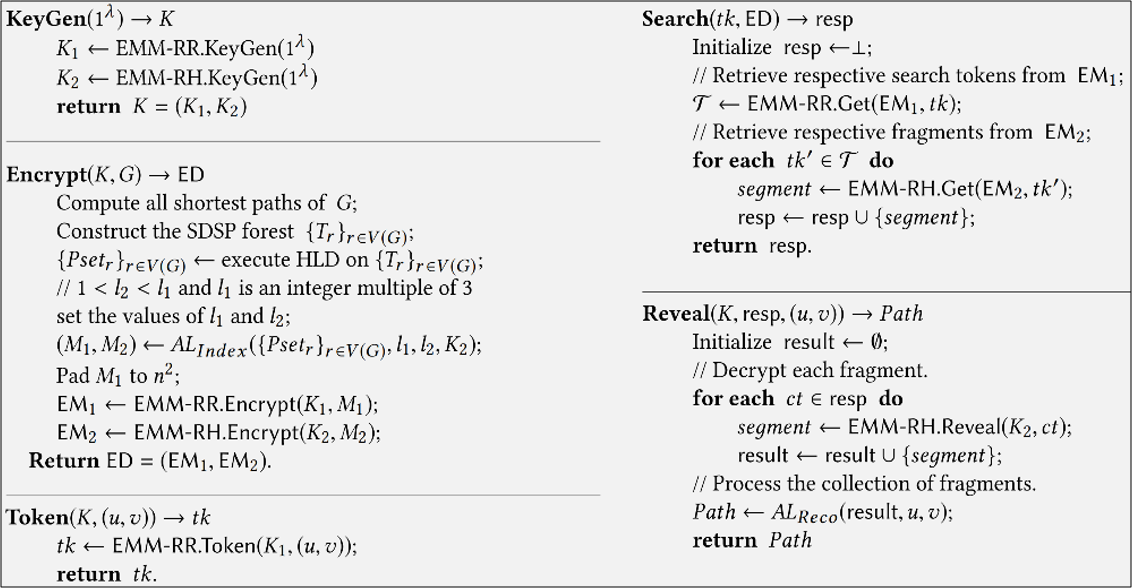}
  \caption{BlindGES Implementation.}
  \label{BlindGESImp}
\end{figure*}

\begin{figure*}[htbp]
  \centering
  \includegraphics[width=0.9\linewidth]{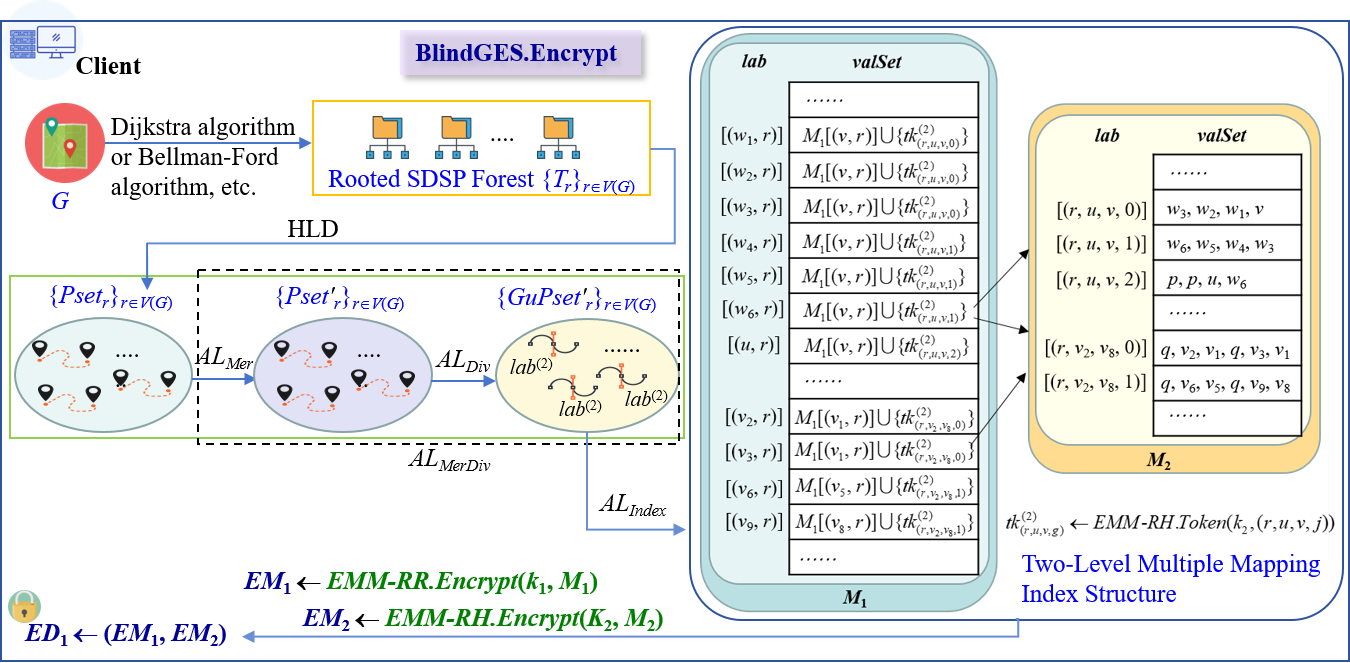}
  \caption{Workflow of setup phase in BlindGES.}
  \label{setup_phase}
\end{figure*}

\subsection{Security Analysis}
\subsubsection{Leakage Functions}
For a given graph $G$, suppose $|G|$ $=$ $n$. For each rooted SDSP tree of $G$, the length of the path is at most $\log n$, the total number of tokens associated with a tree is at most $(n-1)\log n$. Therefore, the total number of entries of $M_1$ is at most $n(n-1)\log n$. In the worst cast, all the edges in each SDSP tree are light edges. After performing HLD, for each tree, the number of paths in $PSet_r$ is at most $n-1$. The lengths of these paths are all $1$. In BlindGES, the lengths of fragments stored in $M_2$ are $l_1$ or $l_2$, where $l_1$ $>$ $l_2$. Therefore, the total number of entries of $M_2$ is at most $\frac{n(n-1)}{l_2}$. To sum up, the upper bounds of the sizes of $M_1$ and $M_2$ are $(n-2)\log n$ and $\frac{n(n-1)}{l_1}$. For the convenience of the proof, we use these upper bounds as $\mathcal{L}^{EMM-RR}_S(M_1)$ and $\mathcal{L}^{EMM-RH}_S(M_2)$, respectively.

Given a multimap, the ($lab$, $valSet$) pairs within it are already determined and the labels in the multimap are all distinct. For $M_1$, one label corresponds to one SPSP query, and different labels correspond to different SPSP queries. Therefore, BlindGES leaks whether two queries are equal (that is $QP^{EMM-RR}_{M_1}$). For $M_2$, one label corresponds one fragment, and different labels correspond to different fragments. Therefore, BlindGES leaks whether two (encrypted) fragments in two responses are equal (that is $QP^{EMM-RH}_{M_2}$). The query pattern (QP) is a $k$ $\times$ $k$ binary matrix $A$ for a sequence of SPSP queries $q_1$, $\cdots$, $q_k$. $A[i][j]$ $=$ $1$ means that the $i$th query was the same as the $j$th query.

Since BlindGES exploits an EMM-RR scheme for $M_1$ and an EMM-RH scheme for $M_2$, BlindGES leaks the leaked information defined by $AP^{EMM-RR}_{M_1}$ and $Vol^{EMM-RH}_{M_2}$. We can use a bipartite graph $BG$ $=$ ($V^{\prime}$, $E^{\prime}$) to describe these leaked information for BlindGES. Given a bipartite graph, its vertices are divided into two disjoint set $V^{\prime}_1$ and $V^{\prime}_2$ such that its every edge connects a vertex from set $V^{\prime}_1$ to a vertex from set $V^{\prime}_2$, and there are no edges within the sets $V^{\prime}_1$ or $V^{\prime}_2$ themselves. The bipartite graph constructed according to $AP^{EMM-RR}_{M_1}$ and $Vol^{EMM-RH}_{M_2}$ is a structure pattern mentioned in \cite{falzon2024PathGES}. We take the label set of $M_1$ as $V^{\prime}_1$ of $BG$ and the label set of $M_2$ as $V^{\prime}_2$ of $BG$. In other words, the vertices in $V^{\prime}_1$ correspond to SPSP queries, and the vertices of $V^{\prime}_2$ correspond to encrypted segments. For a $lab$, $valSet$ pair in $M_1$, $|valSet|$ represents the degree of the vertex corresponding to the label $lab$, which characterizes the number of edges incident to it. The values in the set $valSet$ determine which $|valSet|$ vertices in $V^{\prime}_2$ these $|valSet|$ edges are connected to. Therefore, we can obtain the edge set $E^{\prime}$ of $BG$. Moreover, for a $lab$, $valSet$ pair in $M_2$, since $valSet$ is encrypted, we cannot obtain the values of the set $valSet$ and can only know the size of the ciphertext. The size is taken as the weight of the vertex. There are only two values for the size of all ciphertexts, which correspond to the two grouping sizes $l_1$ and $l_2$ in the segmentation strategy of BlindGES. The degree of a vertex in $V^{\prime}_2$ indicates the number of times the fragment corresponding to the vertex appears in the query response.

In conclusion, the setup leakage function of BlindGES is $\mathcal{L}^{BlindGES}_S(G)$ $=$ $n$, its query leakage function is $\mathcal{L}^{BlindGES}_Q(G)$ $=$ ($A$, $BG$). It should be noted that the above analysis process is reversible. Therefore, $QP^{EMM-RR}_{M_1}$ can be obtained according the query pattern $A$ of $\mathcal{L}^{BlindGES}_Q(G)$, and $AP^{EMM-RR}_{M_1}$, $QP^{EMM-RH}_{M_2}$, and $Vol^{EMM-RH}_{M_2}$ can be obtained according to the structure pattern $BG$ of $\mathcal{L}^{BlindGES}_Q(G)$.

\subsubsection{Security Proof}
\begin{theorem}\label{thm:1}
Suppose \texttt{EMM-RR} to be ($\mathcal{L}^{EMM-RR}_S$, $\mathcal{L}^{EMM-RR}_Q$)-secure and \texttt{EMM-RH} to be ($\mathcal{L}^{EMM-RH}_S$, $\mathcal{L}^{EMM-RH}_Q$)-secure. BlindGES is ($\mathcal{L}_S$, $\mathcal{L}_Q$)-secure according to Definition \ref{Def_1}.
\end{theorem}

\textbf{Proof}:
We denote the simulators of EMM-RR and EMM-RH as $\mathrm{Sim_{RR}}$ and $\mathrm{Sim_{RH}}$, respectively. We need to construct a stateful simulator $\mathrm{Sim_{BlindGES}}$ to simulate the view of the adversary in the Ideal experiment of Definition \ref{Def_1} by using $\mathrm{Sim_{RR}}$ and $\mathrm{Sim_{RH}}$ and prove the view and BlindGES are computationally indistinguishable. The game hops of the hybrid argument are as follows:

\textbf{Hyb 0}: This is exactly like the actual challenger.

\textbf{Hyb 1}: It is the same as the previous Hyb, except for replacing \texttt{EMM-RH.KeyGen} with $\mathrm{Sim_{RH}}(\mathcal{L}^{EMM-RH}_S)$.

If $\mathcal{A}$ can distinguish between the views before and after the replacement, $\mathcal{A}$ can compromise the adaptive security of EMM-RH with a non-negligible advantage.

\textbf{Hyb 2}: It is the same as the previous Hyb, except for replacing \texttt{EMM-RR.KeyGen} with $\mathrm{Sim_{RR}}(\mathcal{L}^{EMM-RR}_S)$.

If $\mathcal{A}$ can distinguish between the views before and after the replacement, $\mathcal{A}$ can compromise the adaptive security of EMM-RR with a non-negligible advantage.

\textbf{Hyb 3}: It is the same as the previous Hyb, except storing some additional state. Construct the bipartite graph $BG$ and the binary matrix $A$ based on the state information of executing the $k$ queries ($q_1$, $\cdots$, $q_k$). The distributions of this Hyb and the previous Hyb are equal.

\textbf{Hyb 4}: It is the same as the previous Hyb, except for invoking the simulators $\mathrm{Sim_{RH}}$ (using $\mathcal{L}^{EMM-RH}_Q$ derived from $\mathcal{L}^{BlindGES}_Q$) instead of the EMM-RH.

If $\mathcal{A}$ can distinguish between the views before and after the replacement, $\mathcal{A}$ can compromise the adaptive security of EMM-RH with a non-negligible advantage.

\textbf{Hyb 5}: It is the same as the previous Hyb, except for invoking the simulators $\mathrm{Sim_{RR}}$ (using $\mathcal{L}^{EMM-RR}_Q$ derived from $\mathcal{L}^{BlindGES}_Q$) instead of the EMM-RR.

If $\mathcal{A}$ can distinguish between the views before and after the replacement, $\mathcal{A}$ can compromise the adaptive security of EMM-RR with a non-negligible advantage.

\textbf{Hyb 6}: It is the same as the previous Hyb, except for invoking the simulators $\mathrm{Sim_{BlindGES}}$ using $\mathcal{L}^{BlindGES}_Q$ instead of $\mathrm{Sim_{RH}}$ and $\mathrm{Sim_{RR}}$. $\mathcal{L}^{BlindGES}_Q$ are equal to the tuple ($\mathcal{L}^{EMM-RH}_Q$, $\mathcal{L}^{EMM-RR}_Q$). Thus this Hyb and the previous Hyb are computationally indistinguishable.

The distribution \textbf{Hyb 6} is identical to that of $\mathrm{Sim_{BlindGES}}$. Therefore, $\mathcal{A}$ can distinguish the views in the Real and Ideal worlds with a negligible advantage, BlindGES is ($\mathcal{L}_S$, $\mathcal{L}_Q$)-secure.

\subsubsection{Cryptanalysis}
\textbf{Additional information leakage caused during the use of the EMM-RH process:} PathGES causes additional information leakage during the process of using a ($\mathcal{L}^{EMM-RH}_S$, $\mathcal{L}^{EMM-RH}_Q$)-secure EMM-RH scheme, where $\mathcal{L}^{EMM-RH}_Q$ $=$ ($QP^{EMM-RH}$, $Vol^{EMM-RH}$), $Vol^{EMM-RH}(M_2, lab)$ $=$ $|M_2[lab]|$. The value of $M_2[lab]$ is the ciphertext of a canonical path fragment. The canonical fragment encoding method used in PathGES makes the length of the path fragments stored in $M_2$ exhibit a regularity, that is, being a power of $2$. Therefore, the size of the corresponding ciphertext also shows a certain regularity. By observing all the ciphertexts in $M_2$ or listening to multiple responses, an adversary can consider the smallest ciphertext size as the ciphertext size corresponding to a path of length $1$. Then, by observing the target response, the adversary can infer the path length of the query path based on the size of the ciphertext contained in the target response. The standardization (i.e., the regularity shown) of path fragments brought about by the canonical fragment encoding makes it possible to conduct side-channel attacks on PathGES. We have verified this through experiments as well. This information leakage is not defined in $\mathcal{L}^{PathGES}_Q$. Therefore, we consider that PathGES has vulnerabilities to side-channel attacks. The Merge-and-Divide mechanism of BlindGES makes the lengths of the fragments stored in $M_2$ be fixed at two lengths, $l_1$ and $l_2$. Therefore, BlindGES will not cause additional information leakage during the use of a EMM-RH scheme.

\textbf{Information leakage caused by the stored fragments:} PathGES pads a path with a length of $l$ to a path of $2^{\lfloor \log n \rfloor + 1}$. Then, $\lfloor \log n \rfloor + 2$ canonical path fragments with lengths of $2^j$ ($j$ $=$ $0$, $1$, $\cdots$, $\lfloor \log n \rfloor + 1$) are extracted, every time starting from the vertex (assumed to be $v$) closest to the root on this path. Suppose the length of this path is $2^{\lfloor \log n \rfloor}$ $+$ $k$, where $1$ $\leq$ $k$ $\leq$ $2^{\lfloor \log n \rfloor + 1}$ $-$ $2^{\lfloor \log n \rfloor}$. The ($2^{\lfloor \log n \rfloor}$ $+$ $j$)-th vertex on this path is $v_j$, ($j$ $=$ $1$, $\cdots$, $k$), and assume that the token corresponding to the canonical path fragment with a length of $2^{\lfloor \log n \rfloor + 1}$ for this path is $tk$. By retrieving $M_1$, it is found that the token sets corresponding to the SPSP queries ($v$, $v_j$) ($j$ $=$ $1$, $\cdots$, $k$) all contain $tk$. For the query ($v$, $v_1$), it obtains additional information: the shortest path information of ($v$, $v_{j^{\prime}}$) ($j^{\prime}$ $=$ $2$, $\cdots$, $k$). PathGES uses the multi-path fragment information contained in the canonical fragments to achieve a one-to-many mapping between tokens and paths. However, the canonical path fragments contain multiple path fragments with a sub-path structure (i.e., short paths are sub-paths of long paths), leading to additional information leakage. The Merge-and-Divide mechanism of BlindGES ensures that the fragments finally stored in $M_2$ are not necessarily all fragments with a sub-path structure, thus reducing the information leakage in the above-mentioned situation. Indeed, this information leakage will not pose an additional threat to the application scenarios where the querier is the graph owner. However, in application scenarios where the querier is not the graph owner, this information leakage will undermine the legitimate rights and interests of the graph owner, so it must be considered. Therefore, BlindGES is more likely to be extended to adapt to new application environments.

\subsection{Empirical Evaluation}
We conducted experiments on a workstation equipped with an NVIDIA GeForce RTX 5060 Ti GPU and an AMD EPYC 7K62 48-core processor (96 threads) with 251 GB of memory. The experiments were implemented using Python version 3.10, along with the following libraries:
\begin{itemize}
    \setlength{\itemsep}{0pt}
    \setlength{\parsep}{0pt}
    \setlength{\parskip}{0pt}
    \item \textbf{NetworkX 3.3}: Used for graph representation and processing. We leveraged NetworkX to construct and manipulate graph structures and compute shortest paths.
    \item \textbf{Cryptography 42.0.8}: Used to perform cryptographic operations. Our scheme adopts AES in CBC mode for symmetric encryption, with both block size and key length set to 16 bytes. Hashing operations use SHA-256, and search tokens are generated using HMAC based on SHA-256.
\end{itemize}

In the experiments, all encryption processes were parallelized to fully utilize the device's 48 physical cores (96 logical threads) of the AMD EPYC 7K62 processor, thereby maximizing computational efficiency. All experiments were conducted on this platform, leveraging Python's parallel computing capabilities to accelerate the encryption procedures. The configuration of the experimental environment and the hardware conditions ensured the reliability and efficiency of the experimental results, providing a solid basis for validating the performance of our comparative evaluations. BlindGES was evaluated using the same social network datasets as those selected by Ghosh et al.\cite{ghosh2021efficient} and Falzon et al. \cite{falzon2024PathGES}, ensuring consistency in the experimental setting. The more details of datasets can see Section \ref{DataSets}. The source code used to implement and evaluate BlindGES is publicly
available at: https://anonymous.4open.science/r/BlindGES/.

\subsubsection{Performance Comparison}
When running the BlindGES scheme, the segment length parameters for the two types of paths are configured as follows: $l_1$ is set to 9, which corresponds to a merged path length of 8, and $l_2$ is set to 5, which corresponds to a non-merged path length of 4. In the setup phase, we used Python's \texttt{os.path.getsize} function to calculate the sizes of the two multimaps, \( M_1 \) and \( M_2 \), and their encrypted data, \( EM_1 \) and \( EM_2 \), for the six datasets. The setup phase also included the time for constructing the multimaps and encryption.

For these six datasets, for the smallest dataset, internetRouting (\(n = 35\)), the total size of the two multimaps \( M_1 \) and \( M_2 \) was 200KB, and the total size of the encrypted data \( EM_1 \) and \( EM_2 \) sent to the server was 504KB, with the setup phase taking 1125 milliseconds. In contrast, for the largest dataset, p2p-Gnutella04 (\(n = 10876\)), the total size of \( M_1 \) and \( M_2 \) reached 37.09GB, and the total size of the encrypted data \( EM_1 \) and \( EM_2 \) was 66.6GB, with the setup phase taking 1.23 hours. In comparison, when running the PathGES scheme under the same conditions, the total size of the two multimaps \( M_1 \) and \( M_2 \) for the internetRouting dataset (\(n = 35\)) was 233KB, and the total size of the encrypted data \( EM_1 \) and \( EM_2 \) was 835KB, with the setup phase taking 1247 milliseconds. For the p2p-Gnutella04 dataset (\(n = 10876\)), the total size of \( M_1 \) and \( M_2 \) increased to 42.08GB, and the total size of the encrypted data \( EM_1 \) and \( EM_2 \) was 97.99GB, with the setup phase taking as long as 2.4 hours. Compared with PathGES, BlindGES demonstrates a clear advantage on complex graphs such as facebook-combined, p2p-Gnutella08, and p2p-Gnutella04. On these three datasets, the setup time of BlindGES is only half that of PathGES, and the total size of the two encrypted databases sent to the server is merely $59\% \sim 67\%$ of that of PathGES. This improvement stems from the fact that PathGES stores a large number of duplicate canonical fragments, resulting in substantial storage redundancy. In contrast, BlindGES adopts a grouping strategy that partitions path fragments into two non-overlapping sets, thereby significantly reducing the storage overhead.The specific details of the two schemes during the setup stage are shown in Table \ref{tab:performance of BlindGES and PathGES in the setup phase}.

\begin{table*}
  \caption{The performance of BlindGES and PathGES in the setup phase.}
  \label{tab:performance of BlindGES and PathGES in the setup phase}
  \footnotesize
  \begin{tabular}{ccccccccc}
    \toprule
    PathGES & Data & $M_1$  & $EM_1$ & $M_2$ & $EM_2$ & time to comp MMs & enc time & setup time \\
    \midrule
    &internetRouting   & 147KB   & 266KB    & 86KB    & 569KB    & 658ms   & 589ms    &1247ms\\
    &Ca-GrQc           & 246KB   & 455KB    & 147KB   & 975KB    & 821ms   & 674ms    &1495ms\\
    &email-EU-core     & 188MB   & 285.5MB  & 84.1MB  & 466.2MB  & 32s     & 55s      & 87s\\
    &facebook-combined & 3.64GB  & 5.15GB   & 1.53GB  & 7.63GB   & 6.8min  & 14.3min  &21.1min\\
    &p2p-Gnutella08    & 9.59GB  & 13.32GB  & 3.76GB  & 18.57GB  & 15.6min & 35.2min  &50.8min\\
    &p2p-Gnutella04    & 30.36GB & 42.37GB  & 11.72GB & 55.62GB  & 42.9min & 100.8mmin &143.7min\\
  \bottomrule
  \toprule
    BlindGES & Data & $M_1$  & $EM_1$ & $M_2$ & $EM_2$ & time to comp MMs & enc time & setup time \\
    \midrule
    &internetRouting   & 147KB   & 266KB     & 53KB   & 238KB   & 599ms   & 526ms   &1125ms\\
    &Ca-GrQc           & 246KB   & 455KB     & 90KB   & 389KB   & 662ms   & 559ms   &1221ms\\
    &email-EU-core     & 188MB   & 285.5MB   & 46.8MB & 195.7MB & 22.8s   & 28.0s   &50.8s\\
    &facebook-combined & 3.64GB  & 5.15GB    & 0.83GB & 2.98GB  & 3.6min  & 6.4min  &10min\\
    &p2p-Gnutella08    & 9.59GB  & 13.32GB   & 2.19GB & 8.02GB  & 9.7min  & 16.3min &26min\\
    &p2p-Gnutella04    & 30.36GB  & 42.37GB    & 6.73GB & 24.23GB & 25.6min & 48.3min &73.9min\\

    \bottomrule
\end{tabular}
\end{table*}

\begin{figure}[htbp]
  \centering
  \begin{subfigure}[b]{0.48\linewidth}
      \centering
      \includegraphics[width=\textwidth]{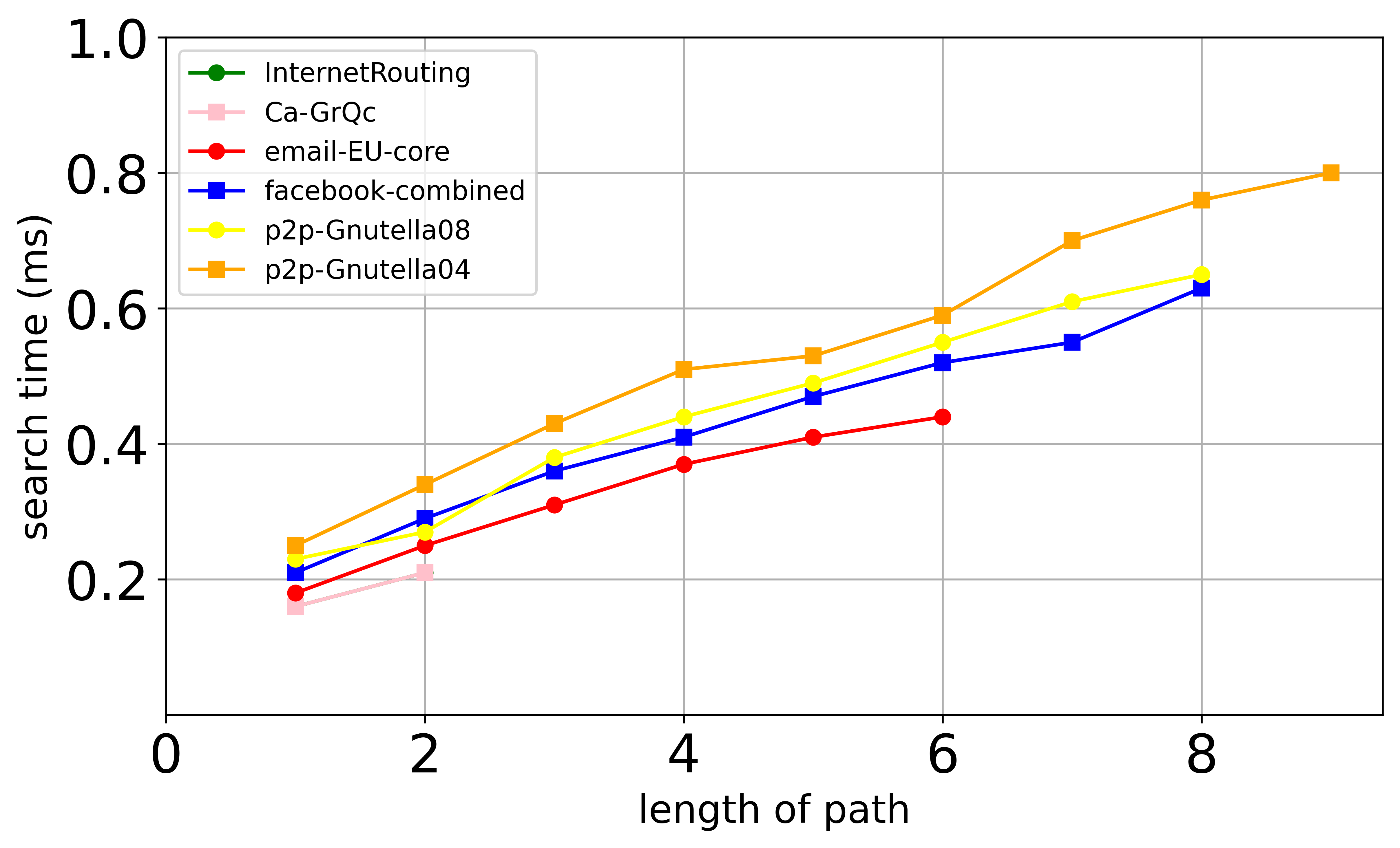}
      \caption{The search time of PathGES }
  \end{subfigure}
  \begin{subfigure}[b]{0.48\linewidth}
      \centering
      \includegraphics[width=\textwidth]{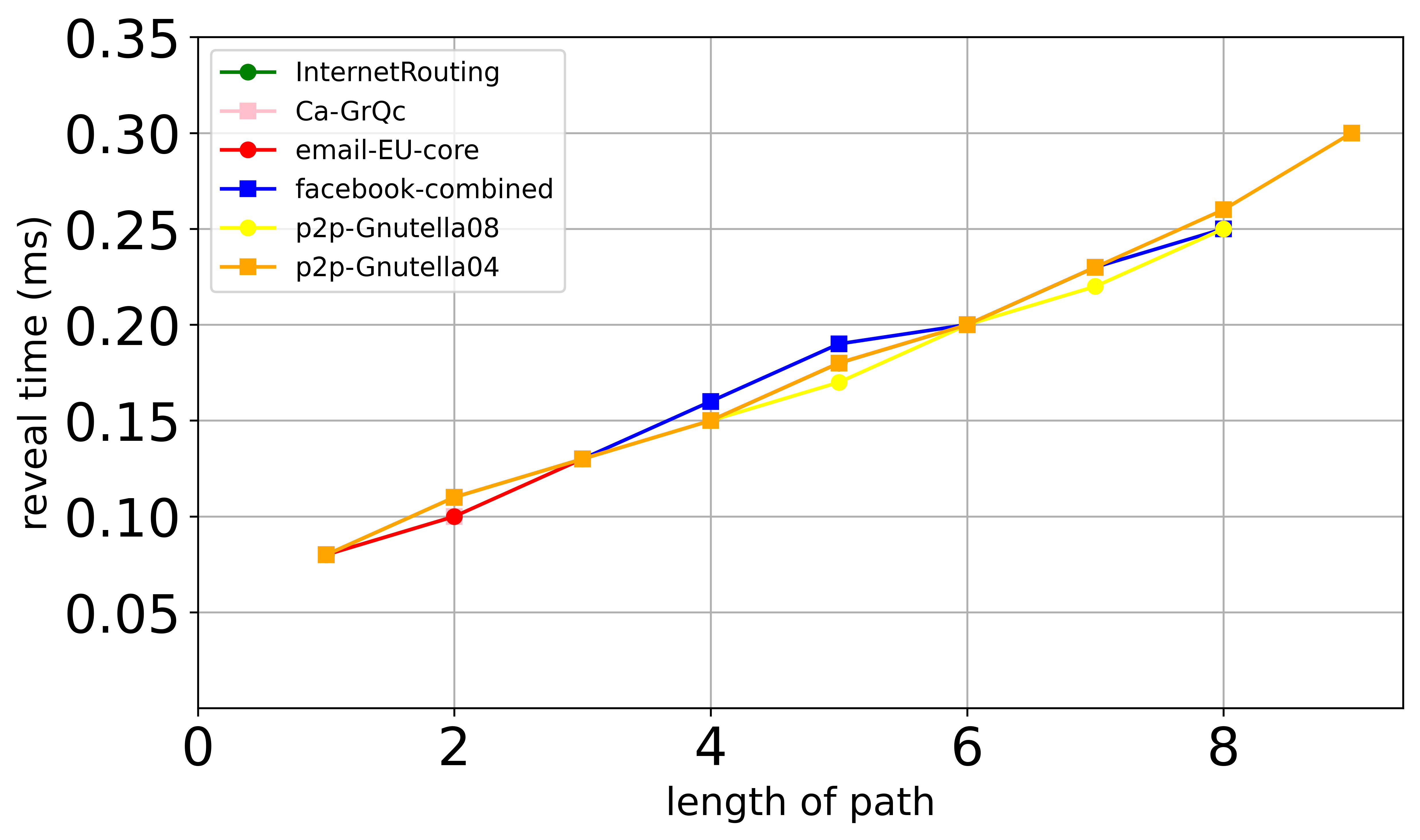}
      \caption{The reveal time of PathGES }
  \end{subfigure}
  \begin{subfigure}[b]{0.48\linewidth}
      \centering
      \includegraphics[width=\textwidth]{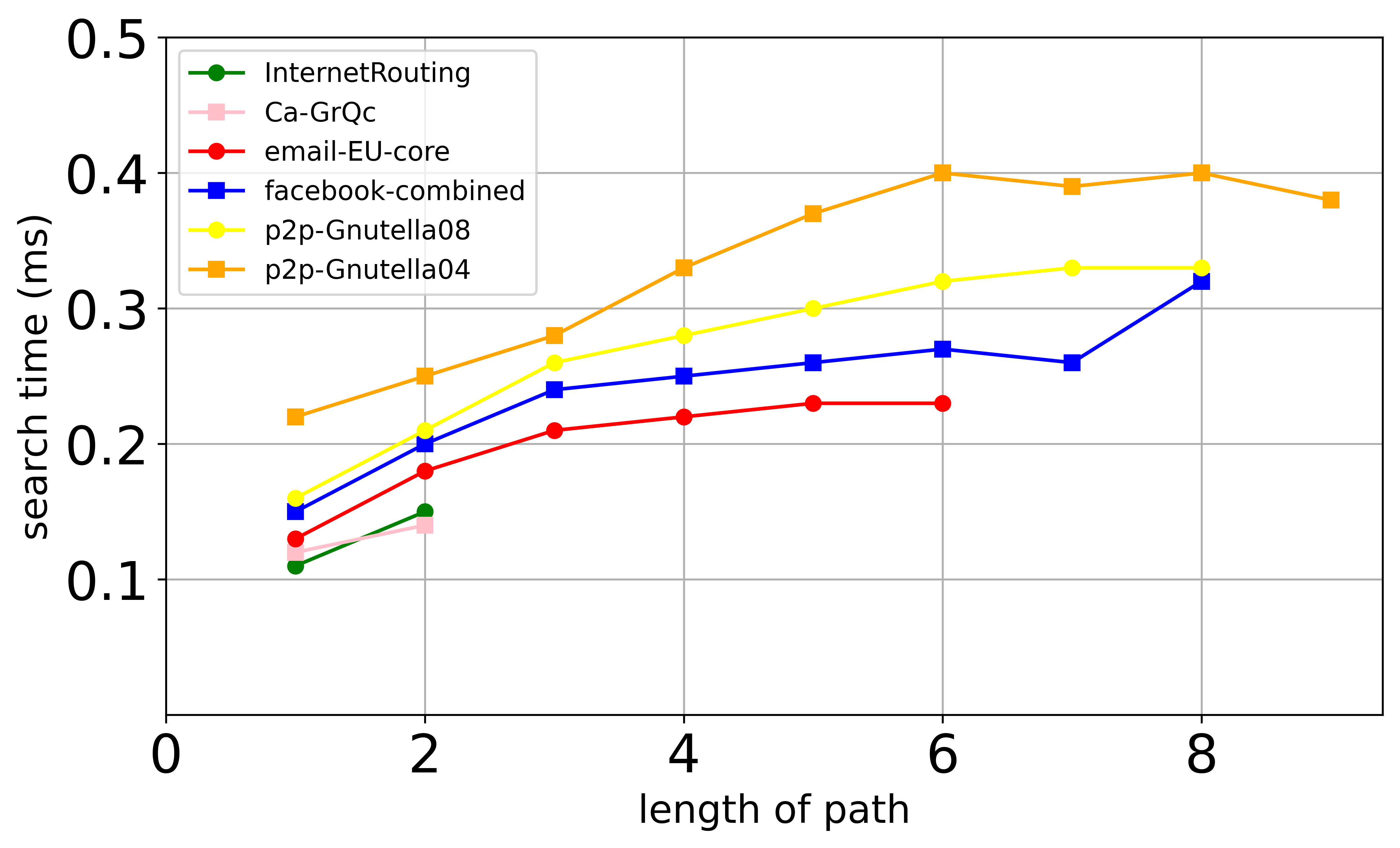}
      \caption{The search time of BlindGES }
  \end{subfigure}
  \begin{subfigure}[b]{0.48\linewidth}
      \centering
      \includegraphics[width=\textwidth]{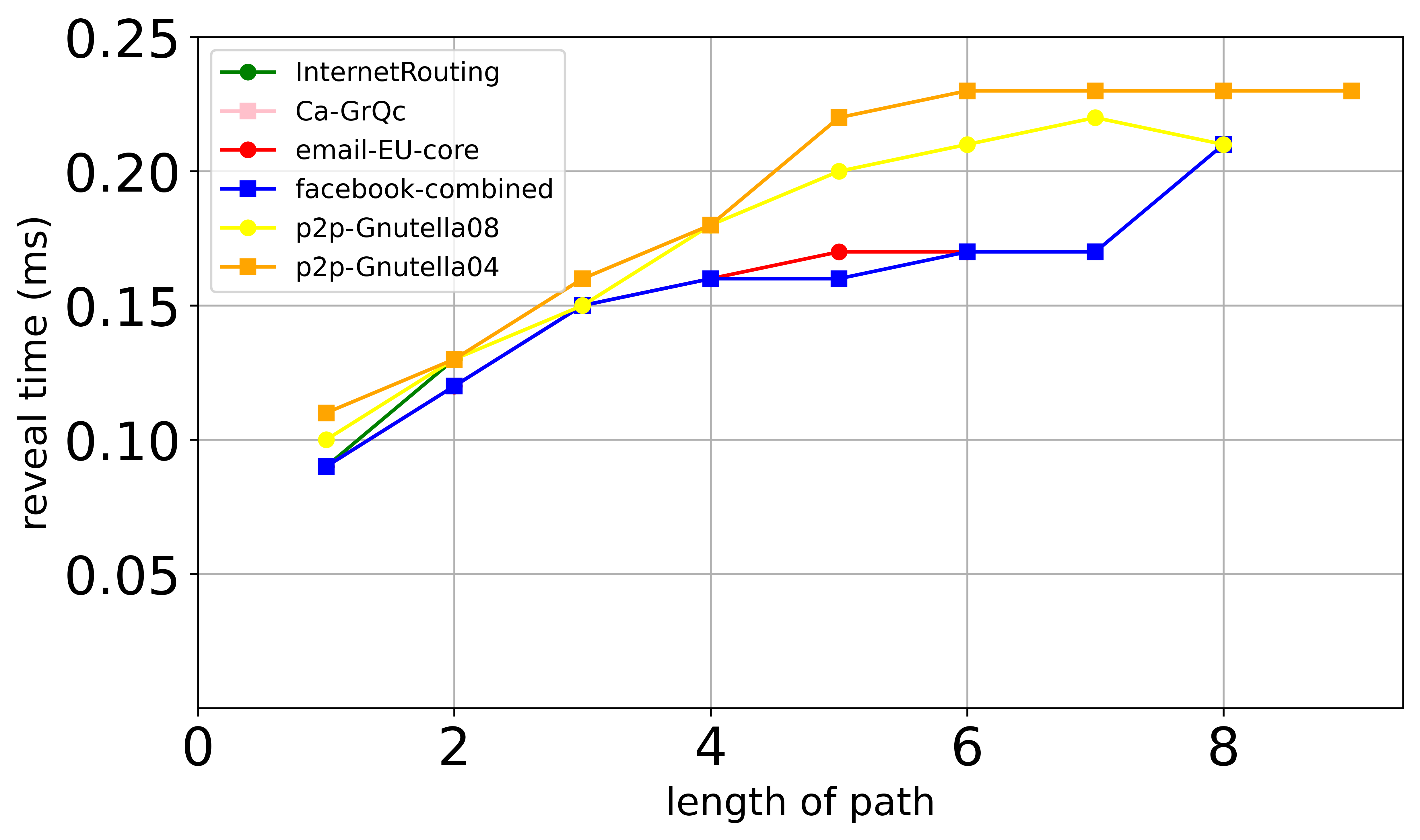}
      \caption{The reveal time of BlindGES }
  \end{subfigure}
  \caption{The search time and reveal time in PathGES and BlindGES}
  \label{Fig_search_time_and_reveal_time}
\end{figure}

\begin{figure}[htbp]
  \centering
  \begin{subfigure}[b]{0.48\linewidth}
      \centering
      \includegraphics[width=\linewidth]{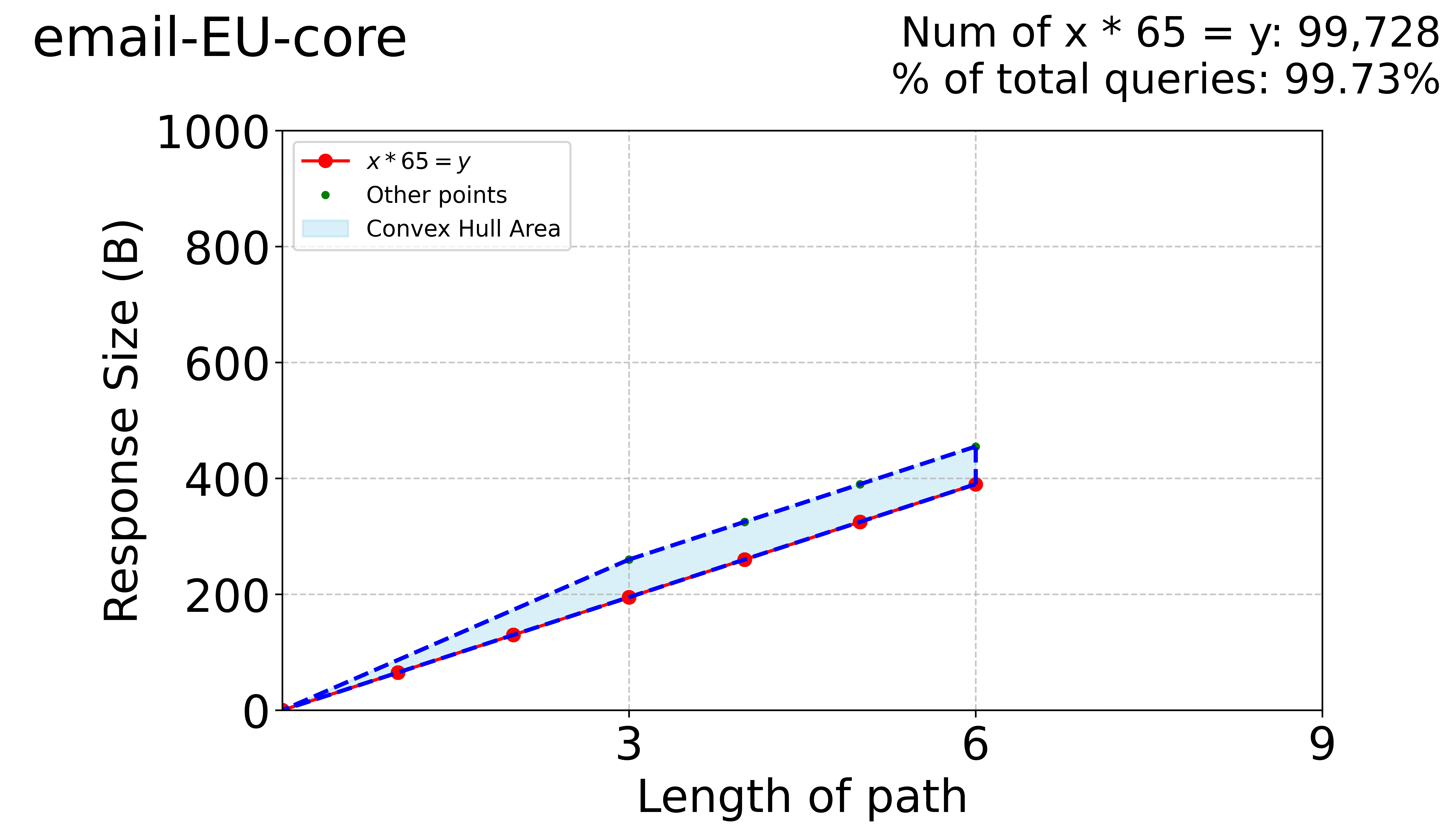}
      \caption{}
  \end{subfigure}
  \begin{subfigure}[b]{0.48\linewidth}
      \centering
      \includegraphics[width=\linewidth]{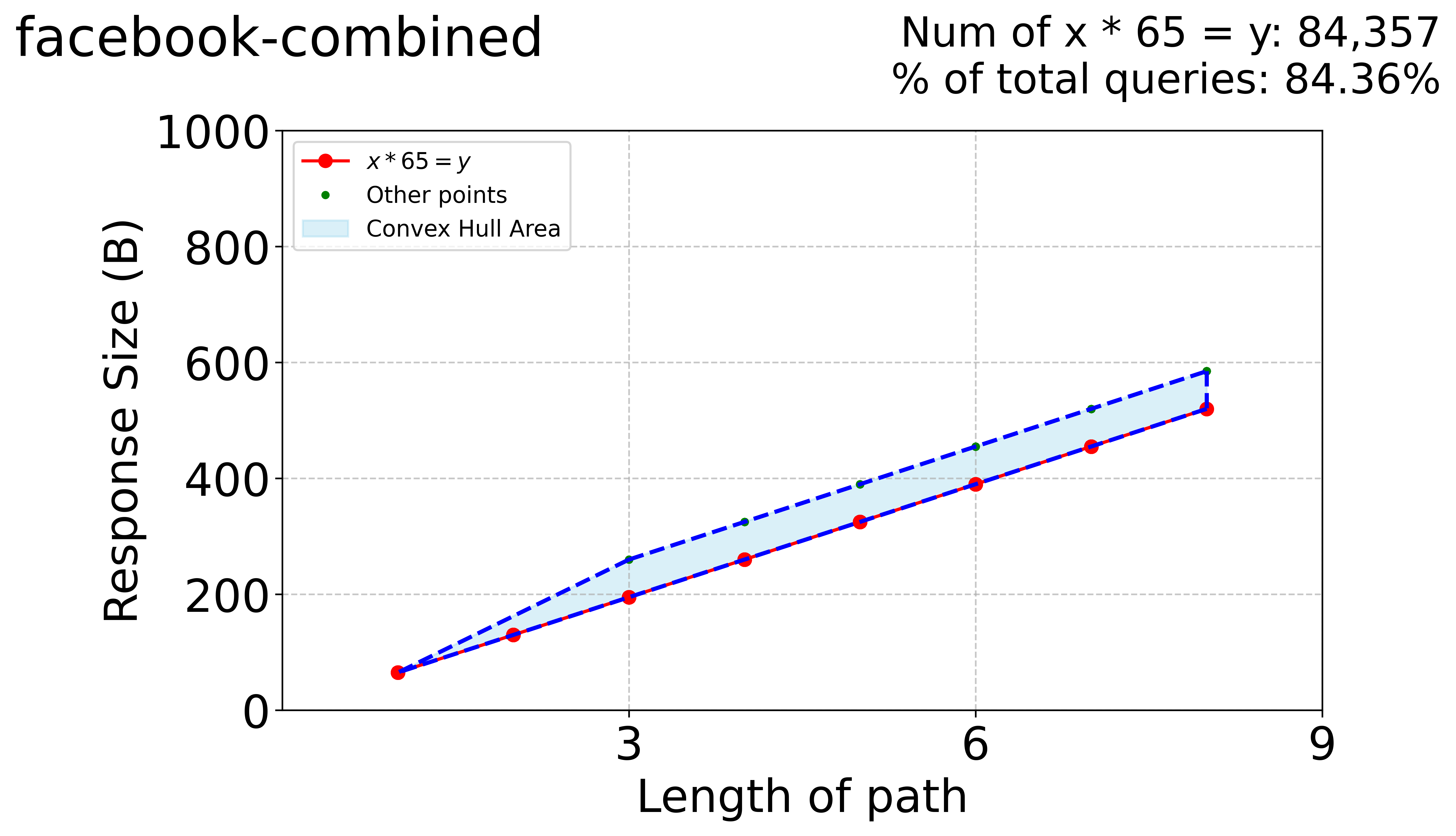}
      \caption{}
  \end{subfigure}
  \begin{subfigure}[b]{0.48\linewidth}
      \centering
      \includegraphics[width=\linewidth]{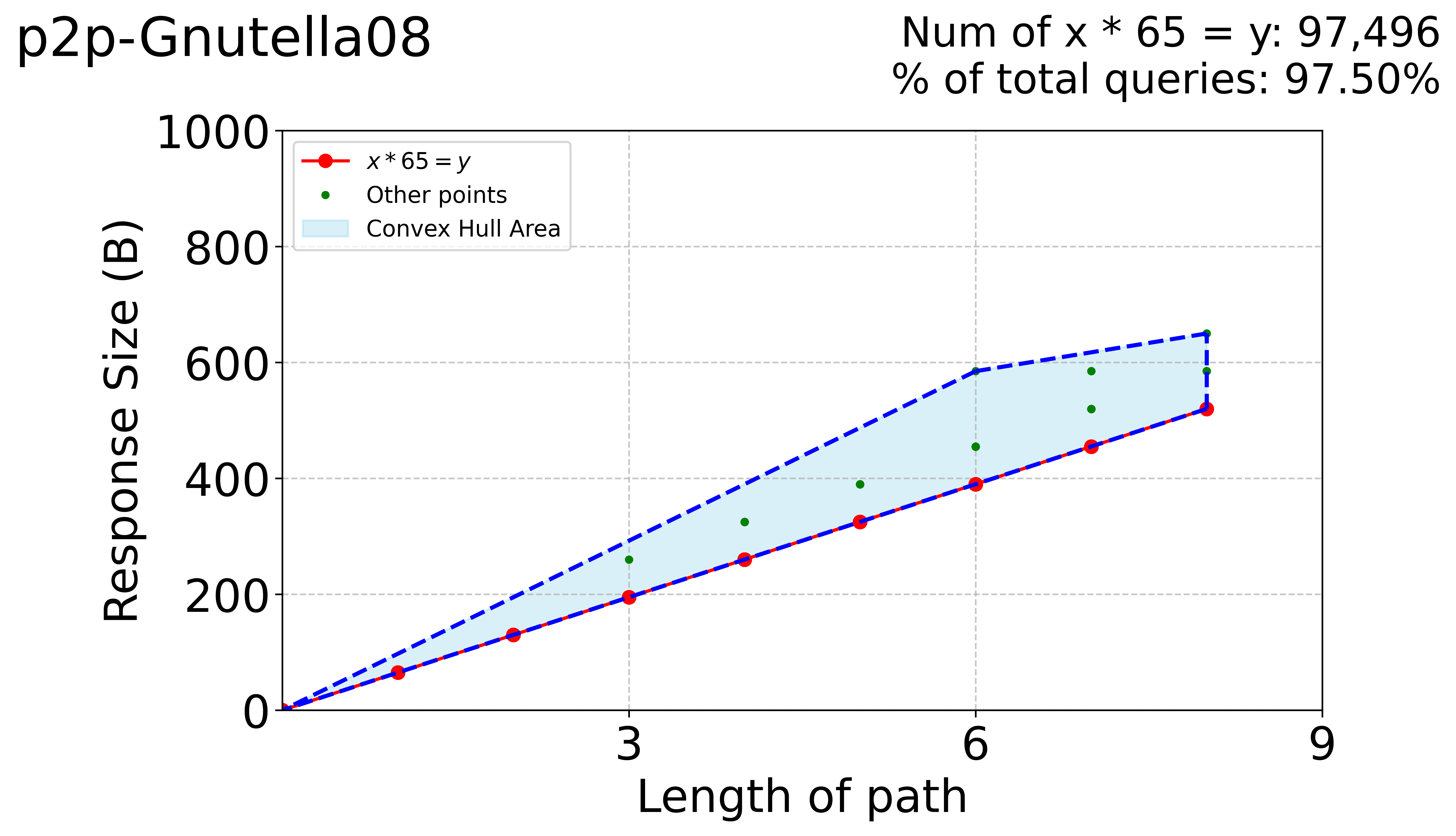}
      \caption{}
  \end{subfigure}
  \begin{subfigure}[b]{0.48\linewidth}
      \centering
      \includegraphics[width=\linewidth]{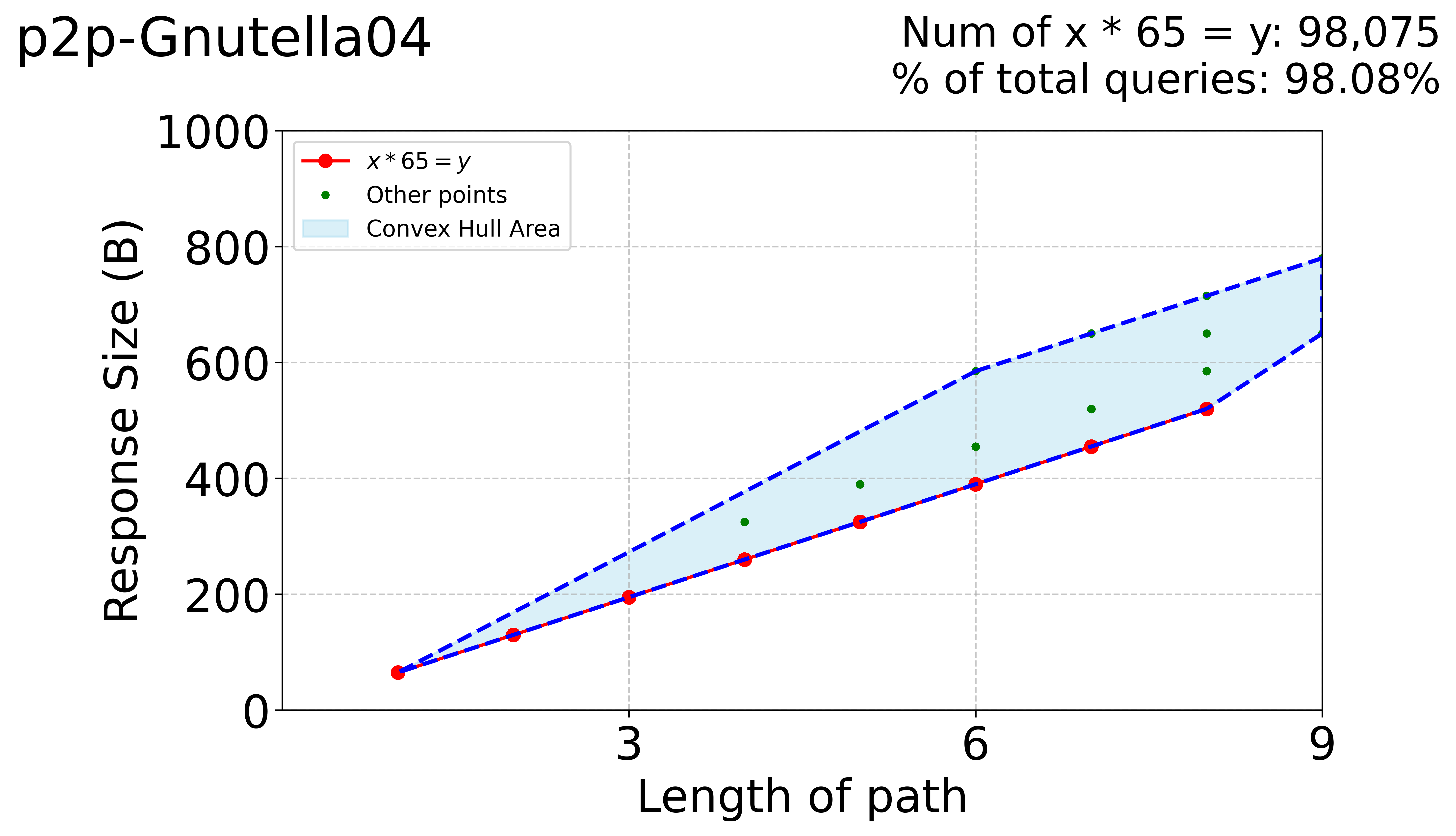}
      \caption{}
  \end{subfigure}
  \caption{Relationship between shortest path length and response size in PathGES. The dots in the shaded area represent possible response values, and points on the line $S = 65 \cdot L$ indicate queries whose path lengths can be directly inferred. }
  \label{Fig resp of path in PathGES}
\end{figure}

We measured the average query response time and reveal time of BlindGES and PathGES for each shortest-path length. As shown in the figure, across the six real-world datasets, both the response time and reveal time of BlindGES remain at the millisecond level. In particular, the query time does not exceed 0.4 milliseconds and is only half that of PathGES, while the reveal time is also slightly faster than that of PathGES.

\subsubsection{Security Comparison}
In this section, we compare BlindGES and PathGES in terms of the relationship between the shortest path length and the response. We argue that queries satisfying \( S = \alpha \cdot L \) will directly expose the length of the shortest path. Here, \( S \) represents the response returned by the query, \( \alpha \) represents the response required for one edge in the shortest path, and \( L \) represents the length of the shortest path. For all responses in the PathGES scheme, the minimum value is \( S_{\text{min}} = 65B \). Based on its multimap construction, an adversary can determine that a shortest path of length 1 requires a response of 65B, meaning that the response for one edge in the path fragment requires \( \alpha = 65 \). If the response size \( S \) returned by the query satisfies \( S = 65 \cdot L \), it implies that the shortest path length of the query can be directly inferred. As shown in Fig.\ref{Fig resp of path in PathGES}, the points in the shaded area represent the possible response values corresponding to the shortest paths of different lengths. , and points on the line \( S = 65 \cdot L \) represent queries where the shortest path length can be inferred. In the statistics of 100,000 random shortest path queries, at least 84\% of the queries' responses are linearly related to the shortest path length. This means that at least 84\% of the queries will directly expose the shortest path length information.

In contrast, our scheme is more effective in reducing the leakage of shortest path information in the response. This is because our scheme hides the response returned for the shortest path of length 1. In the BlindGES parameter setting, the minimum path fragment length mapped by each token is set to 4, and the minimum response \( S_{\text{min}} = 81B \). However, for the adversary, they can only observe that the shortest path fragment requires 81B, without knowing how many edges are included in that fragment. Therefore, when the parameters change, even if the minimum response \( S_{\text{min}} \) is known, the adversary still finds it difficult to accurately deduce the true values of \( \alpha \) and \( L \). As shown in Fig.\ref{Fig resp of path in BlindGES}, The points in the shaded area represent the possible response values corresponding to the shortest paths of different lengths. Points on the line \( S = 81 \cdot L \) represent the inference based on a 81B response being considered as the response for one edge. In this case, fewer than 1\% of the queries are likely to be accidentally guessed in the datasets.

\begin{figure}[htbp]
  \centering
  \begin{subfigure}[b]{0.48\linewidth}
      \centering
      \includegraphics[width=\linewidth]{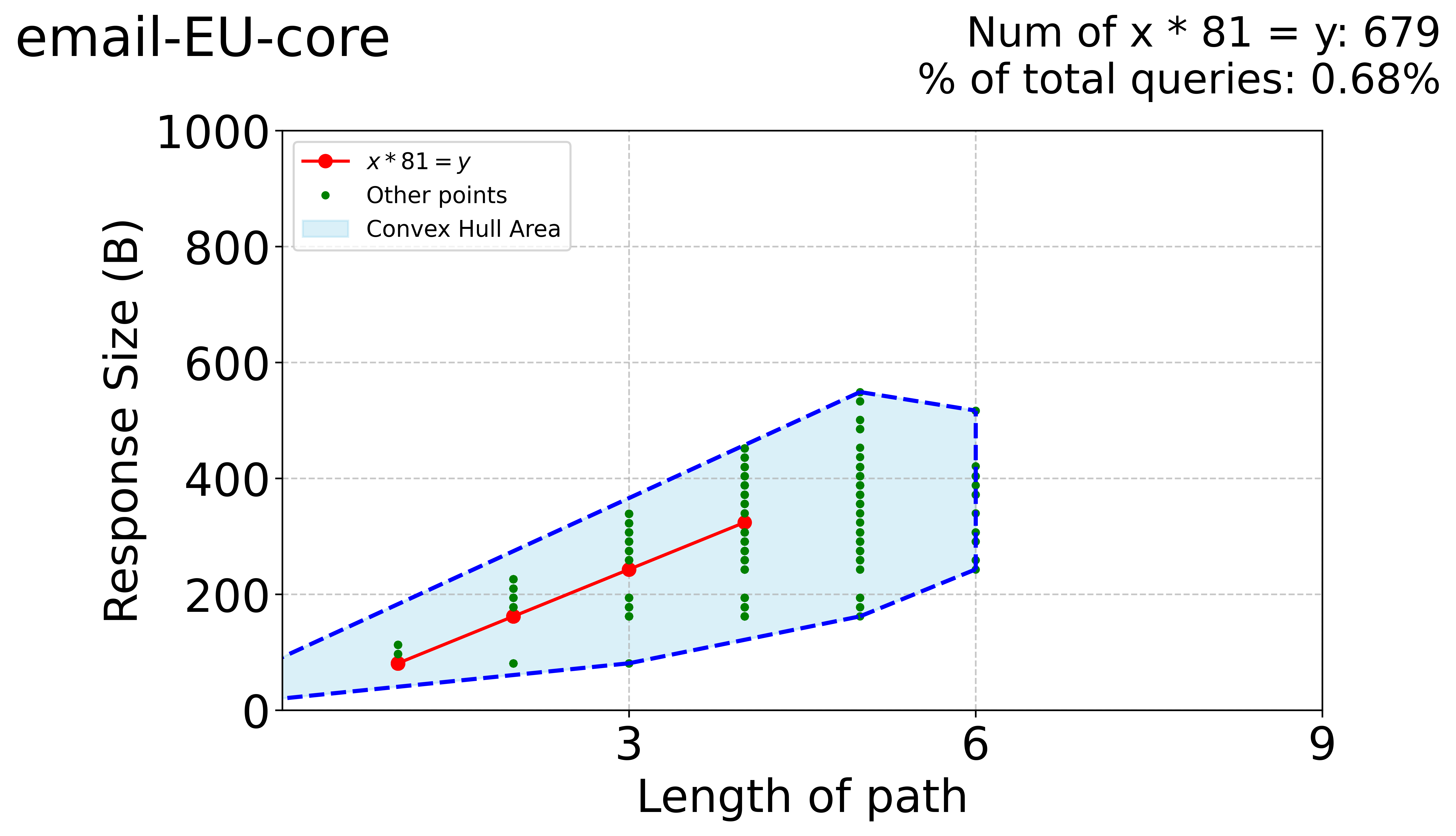}
      \caption{}
  \end{subfigure}
  \begin{subfigure}[b]{0.48\linewidth}
      \centering
      \includegraphics[width=\linewidth]{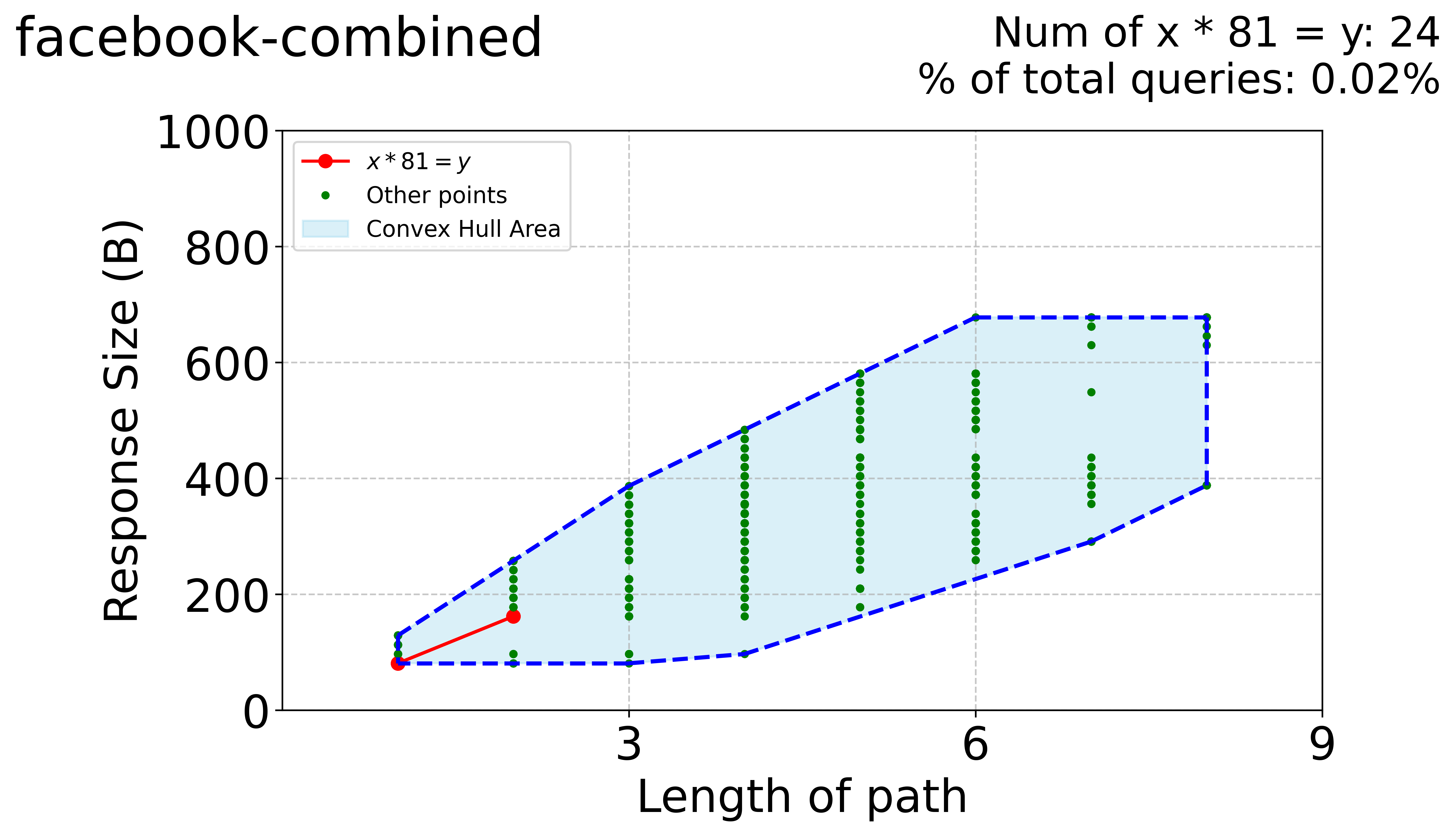}
      \caption{}
  \end{subfigure}
  \begin{subfigure}[b]{0.48\linewidth}
      \centering
      \includegraphics[width=\linewidth]{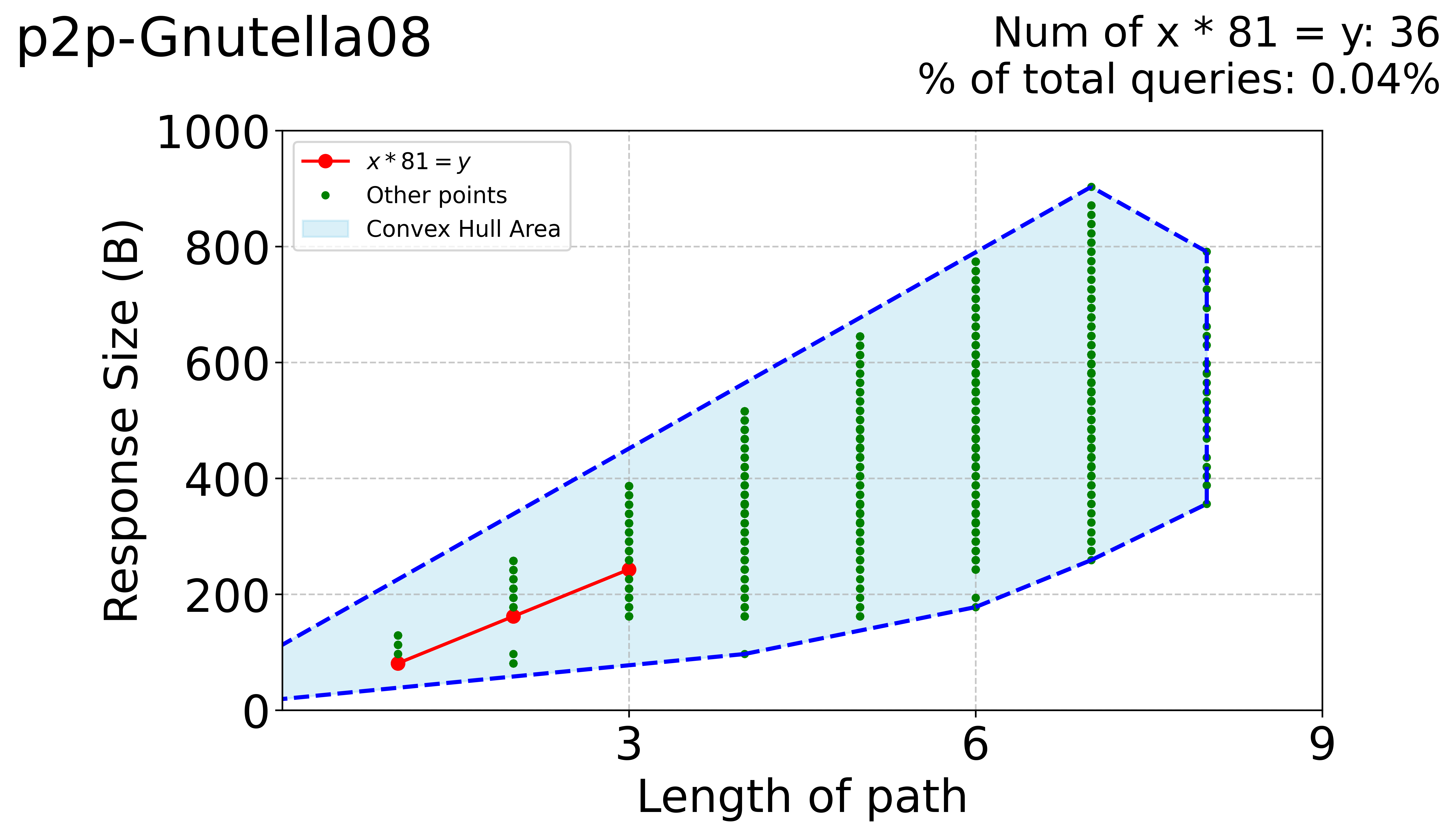}
      \caption{}
  \end{subfigure}
  \begin{subfigure}[b]{0.48\linewidth}
      \centering
      \includegraphics[width=\linewidth]{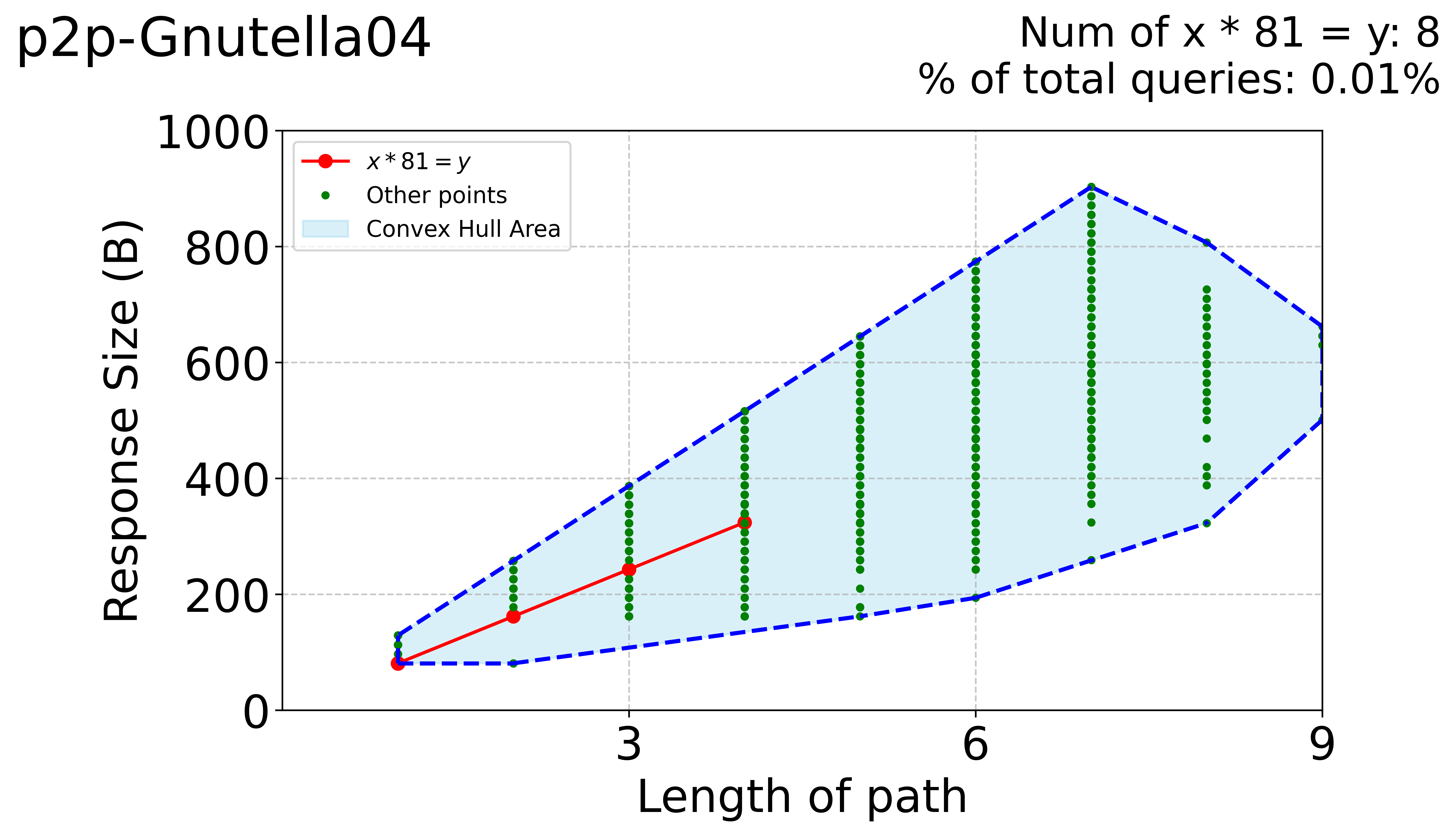}
      \caption{}
  \end{subfigure}
  \caption{Relationship between shortest path length and response size in BlindGES. The dots in the shaded area represent possible response values, and points on the line $S = 81 \cdot L$ indicate queries whose path lengths can be directly inferred.}
  \label{Fig resp of path in BlindGES}
\end{figure}

\section{Limitations and Future Work}
While the Fragment Tree attack demonstrates the effectiveness of exploiting structural leakage in PathGES, and BlindGES provides a robust defense mechanism, both approaches have inherent limitations that point to directions for future research.

\textbf{Limitations of the Fragment Tree Attack.} The attack assumes the server has knowledge of the complete plaintext graph and can observe all query tokens and their relationships over time. In scenarios where the graph is not known to the adversary or where query observation is limited, the attack's effectiveness may be reduced. Additionally, the attack performs best on sparse graphs such as the Gnutella datasets, where approximately 10.24\% of queries can be uniquely recovered. On dense graphs like Ca-GrQc, which has a density of 0.995, the attack achieves only 0.048\% exact recovery because the high connectivity leads to highly homogeneous path names that are difficult to distinguish. The approximate query recovery capability, while narrowing queries to 2-5 candidates for about 10\% of queries on sparse graphs, still requires additional side information to achieve exact recovery in many cases.

\textbf{Limitations of BlindGES.} The Merge-and-Divide mechanism in BlindGES requires the client to predefine two division parameters, $l_1$ and $l_2$  which are set to 9 and 5 respectively in the experiments. These parameters may not be optimal for all graph types, and adaptive parameter selection based on graph density or structure could further improve performance. The setup phase of BlindGES, while significantly faster than PathGES on large graphs, still requires substantial computational resources for very large datasets—on p2p-Gnutella04, the setup time is approximately 73.9 minutes. Additionally, the security analysis assumes the underlying EMM-RR and EMM-RH schemes are secure, and any vulnerabilities in those primitives would affect BlindGES.

\textbf{Future Work Directions.} Several promising directions emerge from this analysis. First, developing adaptive parameter selection mechanisms for BlindGES that automatically tune $l_1$ and $l_2$ based on graph characteristics could improve both security and efficiency across diverse graph types. Second, exploring alternative graph decomposition methods beyond HLD may yield better structural obfuscation for social networks where heavy edges are naturally scarce. Third, integrating differential privacy techniques with graph encryption could provide formal guarantees against side-channel attacks like the response-size inference. Fourth, extending both the Fragment Tree attack and BlindGES to support dynamic graph updates—where vertices and edges are added or removed over time—remains an open challenge, as the current schemes assume static graphs. Fifth, investigating whether the isomorphism-based attack methodology can be applied to other structured encryption schemes beyond PathGES would help understand the broader implications of structural leakage. Finally, developing stronger leakage suppression techniques that go beyond the two-level index structure could further reduce the proportion of one-to-one mappings, ideally approaching zero while maintaining query efficiency

\section{Conclusion}
This report presents a comprehensive analysis of structural leakage in graph encryption schemes for shortest path queries, integrating findings from two complementary works.

PathGES Analysis: We identify that PathGES, despite design intentions, maintains over 99\% one-to-one token-path mappings on real-world datasets due to HLD imbalances. This enables both the Falzon-Paterson attack and side-channel inference of path lengths (84\%+ of queries).

Fragment Tree Attack: We present a novel attack methodology exploiting PathGES structural leakage. The attack constructs fragment trees from HLD decomposition and query trees from leaked tokens, proves their isomorphism, and recovers query contents. Experimental results show up to 10.24\% exact query recovery on sparse graphs and candidate reduction to 2-5 possibilities for additional queries.

BlindGES Defense: We present an enhanced scheme incorporating Merge-and-Divide mechanism and two-level multimap index. BlindGES reduces one-to-one mappings to below 20\%, cuts setup time by 50\%, reduces storage overhead by 32\%, and limits path length leakage to under 1\% while maintaining millisecond-level query response times.

Together, these works demonstrate that structural leakage remains a critical vulnerability in graph encryption, that practical attacks can effectively exploit this leakage, and that careful design with mechanisms like Merge-and-Divide can substantially improve security without sacrificing efficiency.

\bibliographystyle{plain}
\bibliography{reference-SLGE}

\end{document}